\documentclass[11pt]{article}
\usepackage[margin=1.1in]{geometry}
\usepackage{booktabs} 
\usepackage{times}

\usepackage[backref,colorlinks,citecolor=blue,bookmarks=true]{hyperref}
\usepackage{mathtools, amssymb, amsthm, bbm, tabularx}
\usepackage{algorithmic}
\usepackage[ruled,vlined]{algorithm2e}
\usepackage{svg}
 \usepackage{tikz}
\usepackage{float}

\usepackage[title]{appendix}

\usepackage[capitalize]{cleveref}

\theoremstyle{plain}
\newtheorem{theorem}{Theorem}[section]
\newtheorem{lemma}[theorem]{Lemma}
\newtheorem{corollary}[theorem]{Corollary}
\newtheorem{proposition}[theorem]{Proposition}

\theoremstyle{definition}
\newtheorem{definition}[theorem]{Definition}

\theoremstyle{remark}
\newtheorem{remark}[theorem]{Remark}

\numberwithin{equation}{section}

\newenvironment{algorithm2e}{%
  \begin{algorithm}%
}{%
  \end{algorithm}%
}

\usepackage{mathtools, bbm}
\usepackage{enumitem}
\usepackage{booktabs}

\newtheorem{fact}{Fact}

\def\A{\mathcal{A}}

\def\C{\mathcal{C}}
\def\D{\mathcal{D}}

\def\F{\mathcal{F}}

\def\S{\mathbb{S}}

\newcommand*{\N}{{\mathbb{N}}}
\newcommand*{\Z}{{\mathbb{Z}}}
\newcommand*{\R}{{\mathbb{R}}}

\let\eps\epsilon
\let\phi\varphi

\DeclareMathOperator*{\pr}{\mathbb{P}}

\DeclareMathOperator*{\E}{\mathbb{E}}
\DeclareMathOperator*{\var}{Var}

\let\hat\widehat
\DeclareMathOperator{\poly}{poly}

\DeclareMathOperator{\sign}{sign}

\DeclareMathOperator{\proj}{proj}
\DeclareMathOperator{\ind}{\mathbbm{1}}

\newcommand{\cube}[1]{\{\pm 1\}^{#1}}

\newcommand{\ignore}[1]{} 

\newcommand*{\w}{\mathbf{w}}
\newcommand*{\vv}{\mathbf{v}}
\newcommand*{\vu}{\mathbf{u}}

\newcommand*{\x}{\mathbf{x}}
\newcommand*{\z}{\mathbf{z}}

\newcommand*{\Dgeneric}{\D}

\newcommand*{\Dtrain}{\D}
\newcommand{\Dtest}{\D'}

\newcommand{\train}{\mathrm{train}}
\newcommand{\test}{\mathrm{test}}

\newcommand{\concept}{f}
\newcommand{\copt}{\concept^*}

\newcommand{\Slabelled}{S}
\newcommand{\Sunlabelled}{X}
\newcommand{\Strain}{\Slabelled_\train}
\newcommand{\Stest}{\Sunlabelled_\test}

\newcommand{\Gauss}{\mathcal{N}}

\newcommand{\mslack}{\Delta}

\newcommand{\mindex}{\alpha}

\newcommand{\mtrain}{m_{\train}}
\newcommand{\mtest}{m_{\test}}

\newcommand{\T}{\mathcal{T}}

\newcommand{\bias}{\eta}

\newcommand{\convset}{\mathcal{K}}
\newcommand{\subspace}{\mathcal{U}}
\newcommand{\truesubspace}{\mathcal{W}}
\newcommand{\esterror}{\gamma}
\newcommand{\nondegen}{\beta}

\newcommand{\threshold}{T}
\newcommand{\taus}{\mathcal{T}}

\newcommand{\mdegree}{r}

\newcommand{\covar}{M}

\newcommand{\tol}{\phi}
\newcommand{\mrec}{m_{\A}}

\title{Learning Intersections of Halfspaces with Distribution Shift: Improved Algorithms and SQ Lower Bounds}
\author{
     Adam R. Klivans\thanks{\texttt{klivans@cs.utexas.edu}. Supported by NSF award AF-1909204 and the NSF AI Institute for Foundations of Machine Learning (IFML).} \\
	 UT Austin
	 \and Konstantinos Stavropoulos\thanks{\texttt{kstavrop@cs.utexas.edu}. Supported by NSF award AF-1909204, the NSF AI Institute for Foundations of Machine Learning (IFML) and by scholarships from Bodossaki Foundation and Leventis Foundation.} \\
	 UT Austin
     \and
    Arsen Vasilyan\thanks{\texttt{vasilyan@mit.edu}. Supported in part by NSF awards CCF-2006664, DMS-2022448, CCF-1565235, CCF-1955217,\\ CCF-2310818, Big George Fellowship and Fintech@CSAIL. Work done in part while visiting UT Austin.} \\
	 MIT
}

\begin{document}

\maketitle

\begin{abstract}%
  Recent work of Klivans, Stavropoulos, and Vasilyan initiated the study of {\em testable learning with distribution shift} (TDS learning), where a learner is given labeled samples from training distribution ${\cal D}$, unlabeled samples from test distribution ${\cal D'}$, and the goal is to output a classifier with low error on ${\cal D'}$ whenever the training samples pass a corresponding test.  Their model deviates from all prior work in that no assumptions are made on ${\cal D'}$.  Instead, the test must accept (with high probability) when the marginals of the training and test distributions are equal. 

  Here we focus on the fundamental case of intersections of halfspaces with respect to Gaussian training distributions and prove a variety of new upper bounds including a $2^{(k/\eps)^{O(1)}} \mathsf{poly}(d)$-time algorithm for TDS learning intersections of $k$ homogeneous halfspaces to accuracy $\epsilon$ (prior work achieved $d^{(k/\eps)^{O(1)}}$). 
 We work under the mild assumption that the Gaussian training distribution contains at least an $\eps$ fraction of both positive and negative examples ($\eps$-balanced).  We also prove the first set of SQ lower-bounds for any TDS learning problem and show (1) the $\eps$-balanced assumption is necessary for $\mathsf{poly}(d,1/\eps)$-time TDS learning for a single halfspace and (2) a $d^{\tilde{\Omega}(\log 1/\eps)}$ lower bound for the intersection of two general halfspaces, even with the $\eps$-balanced assumption.

 Our techniques significantly expand the toolkit for TDS learning.  We use dimension reduction and coverings to give efficient algorithms for computing a {\em localized} version of discrepancy distance, a key metric from the domain adaptation literature.

\end{abstract}

\vspace{1em}

\section{Introduction}
Distribution shift continues to be a major barrier for deploying AI models, especially in the health and bioscience domains.  By far the most common approach to modeling distribution shift (or domain adaptation) is to bound the performance of a classifier in terms of some notion of distance between the training and test distributions \cite{ben2006analysis,mansour2009domadapt}.  These distances, however, are computationally intractable to estimate, as they are defined in terms of an enumeration over all classifiers from some class.  As such, learners constrained to run in polynomial-time obtain no guarantees on the performance of a classifier (without making strong assumptions on the test distribution).

A recent work of Klivans, Stavropoulos, and Vasilyan \cite{klivans2023testable} departs from this paradigm and defines a model of {\em testable learning with distribution shift} (TDS learning).  In this model, a learner first runs a test on labeled samples drawn from training distribution ${\cal D}$ and {\em unlabeled} samples drawn from test distribution ${\cal D'}$.  No assumptions are made on ${\cal D'}$.  If the test accepts, the learner outputs a classifier that is guaranteed to have low error with respect to ${\cal D'}$.  Further, the test must accept (with high probability) whenever the marginal of ${\cal D}$ equals the marginal of ${\cal D'}$.  It is clear that this model generalizes the traditional PAC model of learning (where ${\cal D}$ always equals ${\cal D'}$), and, as described in \cite{klivans2023testable}, obtaining efficient algorithms seems considerably more challenging.  Giving positive results for TDS learning with running times that match known results in the traditional PAC model is therefore a best-case scenario. 
\subsection{Our Results}
Here we focus on the intensely studied problem of learning intersections of halfspaces (or halfspace intersections) with respect to Gaussian distributions, where large gaps exist between the best known algorithms for TDS learning versus ordinary PAC learning.  Our main contribution is a set of new positive results all of which greatly improve on prior work in TDS learning and in some cases match the best known bounds for PAC learning (see Tables 1 and 2 for precise statements of bounds).  Our algorithm assumes that the training distribution contains at least an $\eps$ fraction of both positive and negative examples ($\eps$-balanced), which turns out to be necessary, as we describe below.

Indeed, we provide the first set of SQ lower bounds for {\em any} problem in TDS learning (that was not already known in the traditional PAC model of learning).  We show that no polynomial-time SQ algorithm can TDS learn a single halfspace unless the training distribution is $\eps$-balanced.  Further, we prove that no polynomial-time SQ algorithm can TDS learn the intersection of two general halfspaces, even if we assume the training distribution is $\eps$-balanced.  Taken together, these results considerably narrow the gap between efficient TDS learnability and PAC learnability for halfspace-based learning. 

\begin{table*}[!htbp]\begin{center}
\begin{tabular}{c c c  c c} 
 \hline
 & \textbf{Type of Intersection} & \textbf{Run-time} &\textbf{Test Set Size} &\textbf{Reference}      \\ \midrule
1 & \begin{tabular}{c} Homogeneous \end{tabular} & $\poly(d) 2^{\poly(\frac{k}{\eps})}$ & $\poly(dk/\eps)$ & \begin{tabular}{c} \Cref{corollary:tds-homogeneous-results-main} \end{tabular} \\ \midrule
2 & \begin{tabular}{c}  Homogeneous \end{tabular} & $(\frac{dk}{\eps})^{O(k)} + d (\frac{k}{\eps})^{O(k^2)}$ & $\poly(dk/\eps)$ & \begin{tabular}{c} \Cref{corollary:tds-homogeneous-results-main} \end{tabular} \\ \midrule 
3 & \begin{tabular}{c}  General  \end{tabular} & $d^{\poly(k/\eps)}$ & $d^{\poly(k/\eps)}$  & \begin{tabular}{c} \cite{klivans2023testable} \end{tabular} \\ \midrule
4 & \begin{tabular}{c}  General  \end{tabular} & \begin{tabular}{c} $d^3 2^{\poly(k/\eps)} +$ \\ $ d^{O(\log(\frac{k}{\eps}))}(\frac{k}{\eps})^{O(k^2)}$ \end{tabular} & $d^{O(\log(\frac{k}{\eps}))}$  & \begin{tabular}{c} \Cref{corollary:tds-general-results-main} \end{tabular} \\ \midrule
5 & \begin{tabular}{c}  Homogeneous \\ Non-Degenerate \end{tabular} & \begin{tabular}{c} $ \poly(d)(\frac{k}{\eps})^{O(k^2)}$ \end{tabular} & $\poly(dk/\eps)$  & \begin{tabular}{c} \Cref{corollary:homogeneous-results} \end{tabular} \\
 \bottomrule
\end{tabular}
\end{center}
\caption{Upper Bounds for TDS Learning $\eps$-Balanced Intersections of $k$ Halfspaces under $\Gauss_d$. All bounds here improve on the best previous bound in row three.  For {\em noise-free PAC learning} intersections of $k$ halfspaces can be learned in time $(dk/\eps)^{O(k)}$ \cite{vempala2010random} and is the best known bound for small $k$. We nearly match this bound in row two above and provide an incomparable result in row four. In row five, we improve on all of these bounds under a non-degeneracy assumption on the intersection of halfspaces; see the Related Work section for a discussion.}
\label{table:upper-bounds}
\end{table*}

\begin{table*}[!htbp]\begin{center}
\begin{tabular}{c c c c} 
 \hline
 & \textbf{Halfspace Type} & \textbf{Assumption on Intersection} &\textbf{SQ Complexity}  \\ \midrule
1 & Homogeneous & Arbitrary  & $\poly(d/\eps)$, for $k=1$ \\ \midrule
2 & Homogeneous & Arbitrary  & $d^{\omega_\eps(1)}$, for $k\ge 2$ \\ \midrule
3 & Homogeneous & $\eps$-Balanced & $\poly(d/\eps)$, for $k=\Theta(1)$  \\ \midrule
4 & General & Arbitrary & $d^{\tilde{\Theta}(\log(1/\eps))}$, for $k= 1$  \\ \midrule
5 & General &  \begin{tabular}{c} $\eps$-Balanced \&\\ $\Theta(1)$-non-degenerate \end{tabular} & $d^{\tilde{\Theta}(\log(\frac{1}{\eps}))}$, for $k\ge 2$, $k=\Theta(1)$ \\ 
 \bottomrule
\end{tabular}
\end{center}
\caption{Statistical Query complexity (upper and lower) bounds for TDS Learning $k$-Halfspace Intersections under $\Gauss_d$. No prior SQ lower bounds for any TDS learning problem were known. For the balance assumption, see \Cref{definition:bounded-bias-non-degeneracy}. For the non-degeneracy assumption, see \Cref{definition:non-degeneracy}. Row 1 and the upper bound of row 4 are from \cite{klivans2023testable}. All other results are from this work: \Cref{theorem:sq-two-homogeneous} (row 2), \Cref{corollary:homogeneous-results} (row 3), \Cref{theorem: TDS learning requires d to the log} (row 4), \Cref{theorem: TDS learning requires d to the log even balanced} (row 5, lower bound),  \Cref{corollary:general-results} (row 5, upper bound). The lower bounds of rows 4, 5 hold for $d=O(\eps^{-1/4})$.}
\label{table:sq-bounds}
\end{table*}


\subsection{Techniques}

\paragraph{TDS Learning through Covering the Solution Space.} Our upper bounds are based on the idea of constructing a set of candidate output hypotheses that has three properties: (1) it has small size, (2) it contains one hypothesis with low test error and (3) all of the hypotheses in the set have low training error. Once such a cover is constructed, a small set of unlabeled data from the test distribution is sufficient to ensure that all of the members of the cover have low training error. This is possible by estimating the discrepancy distance between the test marginal and the Gaussian, but only with respect to the members of the cover, i.e., estimating the maximum probability of disagreement between pairs of elements of the cover under the test marginal. Since the cover is small (by (1)), this can be done efficiently and since all of the hypotheses have low training error (by (3)), the test should accept in the absence of distribution shift. If the test accepts, then all of the members have low disagreement with one hypothesis with low test error (by (2)) and they, hence, have low test error as well. The learner may then output any member of the cover.

\paragraph{Constructing Covers for Halfspace Intersections.} Our method for covering the solution space for TDS learning halfspace intersections is based on two main ingredients. The first ingredient is access to an algorithm that uses training data and retrieves a low-dimensional subspace that is guaranteed to approximately contain (in terms of angular distance) each of the normal vectors that define the ground truth intersection. See the Related Work section for a more detailed discussion on subspace recovery algorithms. The second ingredient is a local halfspace disagreement tester, namely, a tester that takes as input a vector (and unlabelled test data) and certifies that all of the vectors that are geometrically close to the input define halfspaces with low disagreement to the one defined by the input under the test distribution. Such testers have been proposed in the literature of testable learning \cite{gollakota2023efficient,gollakota2023tester} and TDS learning \cite{klivans2023testable}, but, we provide an additional one for the case of general halfspaces. Equipped with both of these ingredients, we use a Euclidean cover for the sphere in the low-dimensional subspace retrieved and run the disagreement tester on each vector in the cover. We form a cover of the solution space with the desired properties by forming all possible intersections of halfspaces with normals in the Euclidean cover and keeping only those with low training error.

For general halfspaces, we also use an additional moment-matching tester which ensures that halfspaces with very high bias can be safely omitted from the output hypothesis, because the test distribution is certified to be sufficiently concentrated in every direction. This is important, because the training data does not reveal enough information for such halfspaces and, hence, it is not guaranteed that their normals will be approximately contained in the retrieved subspace.

\paragraph{SQ Lower Bounds for TDS Learning from Lower Bounds for NGCA.} We prove our statistical query (SQ) lower bounds by reducing appropriate distribution testing problems to TDS learning. The distribution testing problems we consider fall in the category of Non-Gaussian Component Analysis (NGCA) where a distinguisher has access to an unknown distribution and is asked to distringuish whether the distribution is Gaussian or it is Gaussian in all but one hidden direction where the marginal satisfies certain problem-specific conditions. \cite{diakonikolas2023sq} provide SQ lower bounds for various instantiations of the problem. 

We show that a TDS learner for general halfspaces can distinguish the Gaussian from any distribution that has some non-negligible mass far from the origin along some hidden direction. We then construct a distribution that is Gaussian in all but one direction along which the marginal (1) exactly matches moments with the standard Gaussian up to some degree and (2) assigns non-negligible mass far from the origin. To show approximate moment matching, we use a mass transportation argument and for exact moment matching, we use an argument based on the theory of Linear Programming from \cite{pmlr-v195-diakonikolas23b-sq-mixtures}. Under these conditions, a generic tool from \cite{diakonikolas2023sq} implies an SQ lower bound for the distinguishing problem we constructed and hence an SQ lower bound for TDS learning. A similar construction gives a lower bound for intersections of two general halfspaces. For intersections of two homogeneous halfspaces, we reduce the problem of anti-concentration detection (whose SQ lower bound is given in \cite{diakonikolas2023sq}) to the corresponding TDS learning problem.

\subsection{Related Work}
\paragraph{Intersections of Halfspaces} Learning intersections of halfspaces continues to be an important benchmark for algorithm design in learning theory with a long history of prior work \cite{LongW94,BlumK97,KOSOS04,klivans2009cryptographic,klivans2009baumsalgo,KOS08,vempala2010random,vempala2010learning,GKM12,KKM13,DKS18a}.  Finding a fully polynomial-time algorithm for learning the intersection of $k$ halfspaces in $d$ dimensions to accuracy $\epsilon$ remains a notorious open problem, even in the case of noise-free PAC learning with respect to Gaussian marginals. 

The most relevant works here are \cite{vempala2010random} and \cite{vempala2010learning} which both attempt to recover the subspace spanned by the $k$ normals of the relevant halfspaces.  This type of subspace recovery is a crucial ingredient for our work here, as we describe in the Techniques subsection above.    In \cite{vempala2010random}, an algorithm with running time and sample complexity $(dk/\eps)^{O(k)}$ is given for noise-free PAC learning with respect to log-concave marginals.  In a follow-up work \cite{vempala2010learning} claims an improved bound of $(k/\eps)^{O(k)} \mathsf{poly}(d)$.  Unfortunately, this proof has a gap. 
In Appendix \ref{section:appendix-recovery-general} we provide a complete proof of a weaker result using the approach of \cite{vempala2010learning}, namely we obtain a $2^{O(k^2/\eps^2)} \mathsf{poly(d,k)}$ time algorithm for intersections of homogeneous halfspaces.  If we take a non-degeneracy assumption on the ground truth intersection (see Appendix \ref{section:appendix-recovery-non-degenerate}), we prove that the gap can be fixed and we recover the  $(k/\eps)^{O(k)} \mathsf{poly}(d)$ bound. 

For large values of $k$, the best known bound of $d^{\tilde{O}(\log k/\eps^2)}$ for PAC or agnostic learning is due to \cite{KOS08}, obtained using the Gaussian surface area/Hermite analysis approach.  For TDS learning, \cite{klivans2023testable} gave an algorithm with running time $d^{\tilde{O}(k^{6}/\eps^2)}$ that is improper and outputs a polynomial threshold function as the final hypothesis.  In addition to improving their bounds on run-time (as described in Table 1), the algorithm we present here is proper: our learner gives an intersection of $k$ halfspaces as its output hypothesis. 

\paragraph{Distribution Shift/Domain Adaptation}
The field of domain adaptation considers problems very similar to the model introduced here.  A learner is presented with labeled training samples, unlabeled test samples, and is required to output a classifier with low test error.  The learner in traditional domain adaptation, however, is not allowed to reject.  The area is too broad for us to survey here, and we refer the reader to \cite{redko2020survey} and references therein.  We highlight the works of \cite{ben2006analysis} and \cite{mansour2009domadapt}, which provide sample complexity upper bounds for domain adaptation in terms of {\em discrepancy distance}. It is proved in \cite{klivans2023testable} that the notion of discrepancy distance also provides sample complexity guarantees for TDS learning.  The first set of efficient algorithms for domain adaptation without taking strong assumptions on the test distribution were given by \cite{klivans2023testable}.  We also note related work due to \cite{goldwasser2020beyond,kalai2021efficient,goel2023adversarial} on PQ learning, a model formally shown to be harder than TDS learning in \cite{klivans2023testable}. 

\paragraph{Testable Learning}
Although both the Testable Learning framework due to \cite{rubinfeld2022testing} and TDS learning allow a learner to reject unless a training set passes a test, the models address very different issues and are formally incomparable.  In testable learning, the goal is to certify that an {\em agnostic} learner has succeeded (or reject).  In particular, (1) testable learning is trivial in the realizable (noise-free) framework (recall in this paper we work exclusively in a noise-free setting) and (2) testable learning does not allow for distribution shift.  For a further comparison of the models see \cite{klivans2023testable}. We do make use of some general techniques from testable learning, as we describe in the Techniques section.  

\subsection{Preliminaries}

For $\vv\in\R^d,\tau\in \R$, we call a function of the form $\x\mapsto \sign(\vv\cdot \x)$ a homogeneous halfspace and a function of the form $\x\mapsto \sign(\vv\cdot \x + \tau)$ a general halfspace over $\R^d$. An intersection of halfspaces is a function from $\R^d$ to $\cube{}$ of the form $\x\mapsto 2\wedge_{i\in [k]}\ind\{\w^i\cdot \x+\tau^i\ge 0\}-1$, where $\w^i$ are called the normals of the intersection and $\tau^i$ the corresponding thresholds. Let $\Gauss_d$ be the standard Gaussian in $d$ dimensions. For a subspace $\subspace$, let $\proj_\subspace(\w)$ be the orthogonal projection of a vector $\w$ on the subspace $\subspace$.

\paragraph{Learning Setup.} We focus on the framework of \textbf{testable learning with distribution shift (TDS learning)} defined by \cite{klivans2023testable}. In particular, for a concept class $\C\subseteq\{\R^d\to\cube{}\}$, the learner $\A$ is given $\eps,\delta\in(0,1)$, a set $\Strain$ of labelled examples of the form $(\x,\copt(\x))$, where $\x\sim\Dtrain = \Gauss_d$ and $\copt\in\C$, as well a set $\Stest$ of unlabelled examples from an arbitrary test distribution $\Dtest$ and is asked to output a hypothesis $h:\R^d\to \{\pm 1\}$ with the following guarantees.
\begin{enumerate}[label=\textnormal{(}\alph*\textnormal{)}]
    \item (Soundness.) With probability at least $1-\delta$ over the samples $\Strain,\Stest$ we have: 
    
    If $\A$ accepts, then the output $h$ satisfies $\pr_{\x\sim\Dtest}[\copt(\x)\neq h(\x)] \leq \eps$.
    \item (Completeness.) Whenever $\Dtest = \Gauss_d$, $\A$ accepts w.p. at least $1-\delta$ over $\Strain,\Stest$.
\end{enumerate}
If the learner $\A$ enjoys the above guarantees, then $\A$ is called an $(\eps,\delta)$-TDS learner for $\C$ w.r.t. $\Gauss_d$. Since the probability of success can be amplified through repetition (see \cite[Proposition C.1]{klivans2023testable}), in what follows, we will provide algorithms with constant failure probability.

\section{Proper TDS learners for Halfspace Intersections}\label{section:tds-learners}

\subsection{Warm-up: Intersections of Homogeneous Halfspaces}\label{section:algo-homogeneous-intersections}

Our first main result concerns the problem of TDS learning intersections of homogeneous halfspaces with respect to the Gaussian distribution. For a single homogeneous halfspace \cite{klivans2023testable} showed that there is a fully polynomial-time TDS learner under Gaussian marginals. The learner crucially relied on the approximate recovery of the normal vector corresponding to the ground truth halfspace in terms of angular distance using training data. After obtaining a vector that is geometrically close to the ground truth, the learner used unlabelled test data to certify that any halfspace near the recovered one (and, hence, also the ground truth) does not significantly disagree with the recovered halfspace on the test distribution.  Such a certificate can be obtained through appropriate localized testers that rely on low-degree moment estimation (introduced in the testable learning literature, see \cite{gollakota2023efficient,gollakota2023tester}).

We significantly generalize this approach beyond the case of a single halfspace and obtain improved TDS learners for intersections of any number of homogeneous halfspaces (as well as general halfspaces in \Cref{section:algo-general-intersections}). Our approach is once more to recover some information about the ground truth that can be measured in geometric terms. In particular, the appropriate notion of geometric recovery for the case of halfspace intersections is approximate subspace retrieval, namely, recovering a subspace that approximately contains all of the normals to the ground truth intersection, as defined below.

\begin{definition}[Approximate Subspace Retrieval for Homogeneous Halfspaces]\label{definition:subspace-retrieval}
    We say that algorithm $\A$ $(\eps,\delta)$-retrieves the relevant subspace for $\C$ (whose elements are homogeneous halfspace intersections) under $\Gauss_d$ if $\A$, upon receiving at least $\mrec$ examples of the form $(\x,\copt(\x))$, where $\x\sim\Gauss_d$ and $\copt\in\C$, outputs, w.p. at least $1-\delta$ a subspace $\subspace$ such that for any normal $\w$ of $\copt$ we have $\|\proj_\subspace \w\|_2 \ge 1 - \eps$.
\end{definition}

It turns out that the idea of approximate subspace retrieval has been explored in the literature of standard PAC learning, as it can be used to provide strong PAC learning guarantees and proper algorithms. We may, therefore, use existing results on approximate subspace retrieval (see \Cref{section:appendix-pca}) as a first step of our TDS learning algorithm. Once we have obtained a low-dimensional subspace that approximately contains all the normals, we (1) generate a small cover of the candidate solution space, (2) acquire (using unlabeled test examples) a certificate that the cover contains a hypothesis with low test error and (3) bound the discrepancy distance (notion from domain adaptation) of the test marginal with the Gaussian, but only with respect to the candidate solution space. We obtain the following result, whose full proof can be found in \Cref{section:appendix-tds-homogeneous}.

\begin{theorem}[TDS Learning Intersections of Homogeneous Halfspaces]\label{theorem:tds-homogeneous-intersections}
    Let $\C$ be a class whose elements are intersections of $k$ homogeneous halfspaces on $\R^d$, $\eps\in (0,1)$ and $C\ge 1$ a sufficiently large constant. Assume that $\A$ $(\frac{\eps^3}{Ck^3}, 0.01)$-retrieves the relevant subspace for $\C$ under $\Gauss_d$ with sample complexity $\mrec$.
    Then, there is an algorithm (\Cref{algorithm:tds-homogeneous-intersections-appendix}) that $(\eps,\delta=0.02)$-TDS learns the class $\C$, using $\mrec+ \tilde{O}(\frac{dk^2}{\eps^2})$ labeled training examples and $\tilde{O}(\frac{dk^2}{\eps^2})$ unlabelled test examples, calls $\A$ once, and uses additional time $\tilde{O}(\frac{d^3k^2}{\eps^2})+d(k/\eps)^{O(k^2)}$.
\end{theorem}

\begin{algorithm2e}
\caption{Proper TDS Learner for Homogeneous Halfspace Intersections}\label{algorithm:tds-homogeneous-intersections}
\KwIn{Labelled set $\Strain$, unlabelled set $\Stest$, parameter $\epsilon$}
Set $\eps' = \frac{\eps^{3/2}}{Ck^{3/2}}$ and $\eps'' = \frac{\eps^6}{Ck^{7}}$ for some sufficiently large universal constant $C\ge 1$.\\
Run algorithm $\A$ on the set $\Strain$ and let $(\vv^1,\dots,\vv^k)$ be its output. \\
Let $\subspace$ be the subspace spanned by $(\vv^1,\dots,\vv^k)$ and consider the following sparse cover of $\subspace$: 
$\subspace_{\eps''} = \{\frac{\vu}{\|\vu\|_2}: \vu = \eps''\sum_{i=1}^k j_i \vv^i, j_i\in\Z\cap[-\frac{1}{\eps''},\frac{1}{\eps''}],\|\vu\|_2\neq 0\}$ \\
\textbf{Reject} and terminate if $\|\var_{\x\sim\Sunlabelled}(\x)\|_2 \ge 2$. \\
\For{$\vu\in \subspace_{\eps''}$}{
    \textbf{Reject} and terminate if $\pr_{\x\sim\Sunlabelled}[|\vu\cdot \x|\le 2\eps'^{2/3}] > 5\eps'^{2/3}$. \\
} 
Let $\F$ contain the concepts $\concept:\R^d\to\{\pm 1\}$ of the form $\concept(\x)=2\bigwedge_{i=1}^k \ind\{\vu^i\cdot \x \ge 0\} - 1$, where $\vu^1,\dots,\vu^k\in \subspace_{\eps''}$ and $\pr_{(\x,y)\sim \Strain}[y\neq\concept(\x)] \le \eps/5$. \\
\textbf{Reject} and terminate if $\max_{\concept_1,\concept_2\in \F} \pr_{\x\sim \Stest}[\concept_1(\x)\neq \concept_2(\x)] > \eps/2$.\\
\textbf{Otherwise,} output $\hat{\concept}:\R^d\to\cube{}$ for some $\hat\concept\in \F$.
\end{algorithm2e}

Before proving \Cref{theorem:tds-homogeneous-intersections}, we first describe how we can obtain the above algorithm $\A$.

\paragraph{Approximate Subspace Retrieval.} To approximately recover the relevant subspace, we apply results from PAC learning (see \cite{vempala2010learning,vempala2010random}), which we state in \Cref{section:appendix-pca}. For example, \cite{vempala2010learning} uses a Gaussian variance reduction lemma (see \Cref{lemma:variance-reduction}) which states that if we truncate the Gaussian on the positive region of some intersection of homogeneous halfspaces, then the variance of the resulting distribution along the directions that define the normals of the intersection is bounded away below $1$ (for directions orthogonal to the span of the normals, the variance is $1$). Unfortunately, in the original proof of \cite{vempala2010learning}, a (crucial) approximate version of the variance reduction lemma (similar to the last part of \Cref{lemma:variance-reduction}) is missing
and hence it is not clear whether the claimed approximate subspace retrieval result is true. We provide in \Cref{section:appendix-recovery-general,section:appendix-recovery-non-degenerate} a full proof of the subspace retrieval lemma, but with the following caveat: we either (1) incur complexity that is exponential in $\poly(k/\eps)$ (see \Cref{section:appendix-recovery-general}) or (2) require some non-degeneracy assumption (see \Cref{section:appendix-recovery-non-degenerate}).\\

We now give an overview of the proof of \Cref{theorem:tds-homogeneous-intersections}.

\paragraph{Stage I: Acquiring a Good Cover.} A \emph{good cover} is a list $\F$ of candidate hypotheses (i.e., halfspace intersections) that is guaranteed to contain some intersection with low test error \emph{and} only contains intersections with low training error. We construct such a cover as follows.
\begin{enumerate}
    \item Once we have obtained a(n orthonormal basis for a) subspace $\subspace$ such that every normal to the ground truth intersection is geometrically close to some vector in $\subspace$, we exhaustively cover the unit sphere in $\subspace$ (see \Cref{lemma:sparse-cover-angles}) to obtain a list $\subspace'$ of ($(\frac{k}{\eps})^{O(k)}$) candidate unit vectors that is guaranteed to contain, for each normal $\w$ of the ground truth intersection, some element $\vu$, such that the angle between $\w$ and $\vu$ is small.
    \item We then certify that for each element $\vu$ of $\subspace'$, all of the halfspaces whose normals are geometrically close to $\vu$ have low disagreement with the halfspace defined by $\vu$ on the \emph{test distribution}. Such a certificate can be obtained by using tools (\Cref{lemma:homogeneous-disagreement-tester}) from the literature of testable learning (see \cite{gollakota2023efficient,gollakota2023tester}); in fact we may use, here, the same tools that \cite{klivans2023testable} utilized to obtain TDS learners for single homogeneous halfspaces.
    \item We construct $\F$ by including all possible intersections, of at most $k$ elements from $\subspace'$, that have low training error. Note that there is one element $\concept$ in $\F$ such that its normals are (one-by-one) geometrically close to the normals of the ground truth. The previous test has ensured that $\concept$ has low test error, since the probability that any halfspace in $f$ disagrees with the corresponding true one is small.
\end{enumerate}

\paragraph{Stage II: Estimating Discrepancy Distance.} It remains to pick an element from $\F$ with low test error. However, we have only shown that there is one (unknown) element $\concept$ in $\F$ with low test error. Note that since all of the elements of $\F$ have low training error, then the disagreement between each pair of elements in $\F$ should be small under the training marginal (and the test marginal as well if there was no distribution shift). Therefore, as a last step, we test that the disagreement between any pair of hypotheses in $\F$ is small under test data; otherwise, it is safe to reject. If the test accepts, all of the elements in $\F$ should also have low test error (since they mostly agree with $\concept$ under test data). We stress that this last test corresponds to estimating the discrepancy distance  between the test marginal $\Dgeneric'$ and the Gaussian with respect to $\F$, i.e., the quantity
\[
    \mathrm{d}_{\mathrm{disc}}(\Dgeneric',\Gauss ; \F) = \sup_{\concept_1,\concept_2\in\F}\Bigr|\pr_{\x\sim\Dgeneric'}[\concept_1(\x)\neq\concept_2(\x)]-\pr_{\x\sim\Gauss_d}[\concept_1(\x)\neq\concept_2(\x)]\Bigr|
\]
The discrepancy distance is a standard notion in domain adaptation (see, e.g., \cite{mansour2009domadapt}), but involves an enumeration and it can be hard to compute. Since we only compute it with respect to a small set of candidate hypotheses, we can afford to brute force search over all pairs of functions.
Combining our \Cref{theorem:tds-homogeneous-intersections} with tools for approximate subspace retrieval (see \Cref{section:appendix-pca}), we obtain the following upper bounds. For a more detailed version of the bounds, see \Cref{corollary:homogeneous-results}.

\begin{corollary}\label{corollary:tds-homogeneous-results-main}
    The class of $\eps$-balanced intersections of $k$ homogeneous halfspaces on $\R^d$ can be $\eps$-TDS learned in time $\poly(d) 2^{\poly(k/\eps)}$ using $\poly(d) 2^{\poly(k/\eps)}$ training examples and $\poly(dk/\eps)$ test examples. Moreover, it can be $\eps$-TDS learned in time $(\frac{dk}{\eps})^{O(k)}+d(\frac{k}{\eps})^{O(k^2)}$ using $\tilde{O}(d) (\frac{k}{\eps})^{O(k)}$ training examples and $\poly(dk/\eps)$ test examples.
\end{corollary}

\subsection{Intersections of General Halfspaces}\label{section:algo-general-intersections}

In the case of intersections of general halfspaces, we use a similar approach. However, the notion of approximate subspace retrieval of \Cref{definition:subspace-retrieval} is too strong in this case, as there might be halfspaces that have very high bias and, therefore, it is not possible to obtain enough information about them unless we use a vast amount of training data. We, therefore, define the following relaxed version of approximate subspace retrieval, also used for PAC learning (see \cite{vempala2010learning}).

\begin{definition}[Approximate Subspace Retrieval for General Halfspaces]\label{definition:subspace-retrieval-general}
    We say that the algorithm $\A$ $(\eps,\delta,\threshold)$-retrieves the relevant subspace for $\C$ (whose elements are halfspace intersections) under $\Gauss_d$ if $\A$, upon receiving at least $\mrec$ examples of the form $(\x,\copt(\x))$, where $\x\sim\Gauss_d$ and $\copt\in\C$, outputs, w.p. at least $1-\delta$ a subspace $\subspace$ such that for any normal $\w$ corresponding to a halfspace $\{\x:\w\cdot\x+\tau\ge 0\}$ of $\copt$ such that $\tau\le \threshold$, we have $\|\proj_\subspace \w\|_2 \ge 1 - \eps$.
\end{definition}

The notion of approximate subspace retrieval of \Cref{definition:subspace-retrieval-general} is sufficient to design efficient PAC learners, since the halfspaces with large thresholds can be omitted without incurring a significant increase on the error under the training distribution (which, for PAC learning, is the same as the test distribution). In TDS learning, however, the test marginal is allowed to assign non-negligible mass to the unseen region of a hidden halfspace. In fact, this is a source of lower bounds for TDS learning as we show in \Cref{theorem: TDS learning requires d to the log,theorem: TDS learning requires d to the log even balanced}. 

Prior work on TDS learning \cite{klivans2023testable} focusing on the case of a single general halfspace, used a moment matching tester to ensure that the test marginal does not assign considerable mass to the unseen region of significantly biased halfspaces (as is the case under the Gaussian). Such tests incur a complexity of $d^{\Theta(\log(\frac{1}{\eps}))}$, which is essentially unavoidable (see \Cref{theorem: TDS learning requires d to the log}). Note that by assuming that the ground truth is balanced (\Cref{definition:bounded-bias-non-degeneracy}), one can bypass the lower bound of \Cref{theorem: TDS learning requires d to the log} for TDS learning a single general halfspace. This is not the case, however, for intersections of even $2$ general halfspaces (see \Cref{theorem: TDS learning requires d to the log even balanced}), where the lower bound of $d^{\tilde{\Omega}(\log(1/\eps))}$ persists even under the balanced concepts assumption.

For TDS learning general halfspaces, we adopt a similar moment matching approach as the one used for a single general halfspace (see \cite{klivans2023testable}) to ensure that the normals of the ground truth that are not represented by any element of the retrieved subspace (due to high bias) are not important even under the test distribution. Moreover, in order to acquire a certificate that we have a good cover (as per the previous section), we design a local halfspace disagreement tester that works even for general halfspaces (see \Cref{lemma:general-disagreement-tester}). We obtain the following result (see \Cref{section:appendix-tds-general}).

\begin{theorem}[TDS Learning Intersections of General Halfspaces]\label{theorem:tds-general-intersections}
    Let $\C$ be a class whose elements are intersections of $k$ general halfspaces on $\R^d$, $\eps\in (0,1)$ and $C\ge 1$ a sufficiently large constant. Assume that $\A$ $(\frac{\eps^3}{Ck^3}, 0.01,3\log^{1/2}(\frac{10k}{\eps}))$-retrieves the relevant subspace for $\C$ under $\Gauss_d$ with sample complexity $\mrec$.
    Then, there is an algorithm (\Cref{algorithm:tds-general-intersections-appendix}) that $(\eps,\delta=0.02)$-TDS learns the class $\C$, using $\mrec+ \tilde{O}(\frac{dk^2}{\eps^2})$ labelled training examples and $d^{O(\log(k/\eps))}$ unlabelled test examples, calls $\A$ once and uses additional time $d^{O(\log(k/\eps))}(k/\eps)^{O(k^2)}$.
\end{theorem}

We once more combine our \Cref{theorem:tds-general-intersections} with results on approximate subspace retrieval (see \Cref{section:appendix-pca}), to obtain the following upper bounds (see also \Cref{corollary:general-results}).

\begin{corollary}\label{corollary:tds-general-results-main}
    The class of $\eps$-balanced intersections of $k$ general halfspaces on $\R^d$ can be $\eps$-TDS learned in time $d^3 2^{\poly(k/\eps)} + d^{O(\log(k/\eps))} (k/\eps)^{O(k^2)}$ using $ \tilde{O}(d) 2^{\poly(k/\eps)}$ training examples and $d^{O(\log(k/\eps))}$ test examples.
\end{corollary}

\section{Statistical Query Lower Bounds}\label{section:sq-lower-bounds}

We will now provide a number of lower bounds for TDS learning in the statistical query model originally defined by \cite{kearns1998efficient}, which has been a standard framework for proving computational lower bounds in machine learning, and is known to capture most commonly used algorithmic techniques like gradient descent, moment methods, etc. (see, for example, \cite{feldman2017statistical,feldman2017statisticalquery}.

\begin{definition}[Statistical Query Model]\label{definition:sq}
    Let $\tol>0$ and $\Dgeneric$ be a distribution over $\R^d$. We say that an algorithm $\A$ is a statistical query algorithm (SQ algorithm) with tolerance $\tol$ if $\A$ only has access to $\Dgeneric$ through making a number of (adaptive) bounded queries of the form $q:\R^d\to[-1,1]$, for each of which it receives a value $v\in\R$ with $|v-\E_{\x\sim\Dgeneric}[q(\x)]| \le \tol$.
\end{definition}

Our approach is to reduce appropriate distribution testing problems to TDS learning and then show that these problems cannot be efficiently solved in the SQ framework, by applying recent results from \cite{diakonikolas2023sq} on Non-Gaussian Component Analysis. 

\subsection{General Halfspaces: A Tight Lower Bound}\label{section:sq-general-halfspaces}

We prove the following theorem which gives a tight lower bound for TDS learning general halfspaces with respect to the Gaussian distribution in the SQ framework, since the lower bound matches the recent corresponding upper bound of \cite{klivans2023testable}.

\begin{theorem}[SQ Lower Bound for TDS Learning a Single Halfspace]
\label{theorem: TDS learning requires d to the log}
For $\epsilon>0$, set $d = \eps^{-1/4}$. Then, for all sufficiently small $\epsilon$, the following is true. Let $\mathcal{A}$ be a TDS learning algorithm for general halfspaces over $\R^{d}$ w.r.t. $\Gauss_d$, with accuracy parameter $\epsilon$ and success probability at least
$0.95$. Further, suppose that $\mathcal{A}$ obtains at most $d^{\frac{\log1/\epsilon}{\log\log1/\epsilon}}$ samples from the training distribution and accesses the testing distribution via $2^{d^{o(1)}}$ SQ queries of precision $\tol>0$ (the SQ queries
are allowed to depend on the training samples). Then, the tolerance $\tol$ has to be at most $d^{-\Omega(\frac{\log1/\epsilon}{\log\log1/\epsilon})}$.
\end{theorem}

We first define an appropriate distribution testing problem which can be reduced to TDS learning general halfspaces. In particular, the distribution testing problem we define amounts to testing whether a distribution to which we have sample access assigns too much mass to some halfspace compare to the mass assigned by the Gaussian.

\begin{definition}[Biased Halfspace Detection Problem]\label{definition:bias-halfspace-detection}
    Let $0\le \alpha\le \beta\le 1$ The $(\alpha,\beta)$-biased halfspace detection problem is the task
of distinguishing the $d$-dimensional standard Gaussian distribution
from any distribution $\Dgeneric$ over $\R^{d}$ for which there exist $\vv$
in $\R^{d}$ and $\tau$ in $\R$ satisfying 
\begin{align*}
\pr_{\x\sim \Dgeneric}[\x\cdot\vv\geq \tau] \geq\beta \;\;\text{ and }\;\; \pr_{\x\sim\Gauss_d}[\x\cdot\vv\geq \tau] \leq\alpha
\end{align*}
\end{definition}

The idea is that if one has a TDS learner for general halfspaces, then the TDS learner must also work when the training examples are drawn from a Gaussian and labelled by the constant hypothesis $-1$. In this case, the learner cannot extract any information about the training data, except from the fact that they correspond to a halfspace with very high bias (but the direction remains completely unspecified). If the test distribution assigns a lot of mass on the positive region of the halfspace, then the error would be large and the TDS learner will reject. On the other hand, if the test distribution is the Gaussian, the TDS learner will accept. Hence, the TDS learner would solve the biased halfspace detection problem. We obtain the following quantitative result, whose formal proof can be found in \Cref{section:appendix-sq-single-halfspaces}.

\begin{proposition}[Biased Halfspace Detection via TDS Learning]
\label{proposition: tds learning implies biased halfspace detection}
Let $\mathcal{A}$ be a TDS learning algorithm for general halfspaces
over $\R^{d}$ w.r.t. $\Gauss_d$ with accuracy parameter
$\eps$ and success probability at least $0.95$. Suppose $\mathcal{A}$
obtains at most $m$ samples from the training distribution and accesses
the test distribution via $N$ SQ queries of tolerance $\tol$
(the SQ queries are allowed to depend on the training samples). Then,
there exists an algorithm $(\frac{1}{100m},10\epsilon)$-biased halfspace
detection that uses $N+1$ SQ queries of tolerance $\min\left(\tol,\epsilon\right)$
and has success probability at least $0.8$. 
\end{proposition}

In order to complete the proof of \Cref{theorem: TDS learning requires d to the log}, it remains to show that the biased halfspace detection problem is hard in the SQ framework. To this end, we use a powerful tool from recent work on Non-Gaussian Component Analysis by \cite{diakonikolas2023sq}, which states that distinguishing the Gaussian from a distribution which is Gaussian in all but one hidden direction is hard for SQ algorithms, whenever the marginal in this direction is guaranteed to match the low degree moments of the Gaussian (see \Cref{theorem: moment matching implies lowe bound}). For our purposes, it is sufficient to construct a one-dimensional distribution that matches low degree moments with the standard Gaussian, but assigns non negligible mass far from the origin. We obtain the following result whose proof can be found in \Cref{section:appendix-sq-single-halfspaces}.

\begin{proposition}[SQ Lower Bound for Biased Halfspace Detection]
\label{proposition: detecting biased halfspaces is hard} For $\epsilon>0$,
set $d=\frac{1}{\epsilon^{1/4}}$. Then, for all
sufficiently small $\epsilon$, the following is true. Suppose that $\mathcal{A}$
is an SQ algorithm for $(d^{-\ln(1/\epsilon)},10\epsilon)$-biased
halfspace detection problem over $\R^{d}$, and $\mathcal{A}$ has
a success probability of at least $2/3$. Then, $\mathcal{A}$ either
has to use SQ tolerance of $d^{-\Omega(\frac{\log1/\epsilon}{\log\log1/\epsilon})}$,
or make $2^{d^{\Omega(1)}}$ SQ queries.
\end{proposition}

\subsection{Intersections of Two Homogeneous Halfspaces}\label{section:sq-intersections-homogeneous}

The following theorem demonstrates that, although TDS learning a single homogeneous halfspace with respect to the Gaussian distribution admits fully polynomial time algorithms (see \cite{klivans2023testable}), for intersections of two homogeneous halfspaces, there is no polynomial-time SQ algorithm. Notably, the construction corresponds to a highly unbalanced intersection, so the lower bound does not hold for the problem of TDS learning balanced intersections.

\begin{theorem}[SQ Lower Bound for TDS Learning Two Homogeneous Halfspaces]\label{theorem:sq-two-homogeneous}
    Let $\eps>0$ with $\eps\in(0,1/10)$ and let $\A$ be a TDS learning algorithm for learning intersections of $2$ homogeneous halfspaces over $\R^d$ w.r.t. $\Gauss_d$ with accuracy $\eps$ and success probability at least $0.95$. Then $\A$ either makes some query of tolerance $\tol = d^{-\omega_\eps(1)}$ to the test distribution or runs in time $d^{\omega_\eps(1)}$.
\end{theorem}

To prove our result, we use an SQ lower bound for detecting anti-concentration (AC) from \cite{diakonikolas2023sq}.

\begin{theorem}[SQ Lower Bound for Detecting AC, Theorem 1.10 in \cite{diakonikolas2023sq}]\label{theorem:sq-hardness-ac}
    Let $\eps \in (0,1/2)$. Any SQ algorithm with SQ access to either (1) $\Gauss_d$ or (2) some distribution $\Dgeneric'$ that assigns mass at least $\eps$ on some subspace of dimension $d-1$ and distinguishes the two cases w.p. at least $2/3$, either uses $2^{d^{\Omega(1)}}$ queries, or uses a query with tolerance at most $d^{-\omega_\eps(1)}$.
\end{theorem}

It remains to reduce the AC detection problem to the problem of TDS learning intersections of two homogeneous halfspaces. The idea is to use an intersection of two almost opposite halfspaces, whose positive region effectively coincides with half of the subspace where $\Dgeneric'$ has non negligible mass. Therefore, upon acceptance, the output function should take the value $1$ with non-negligible probability only if the unknown distribution is $\Dgeneric'$, which implies that we have solved the distinguishing problem. See \Cref{section:appendix-hardness-two-homogeneous} for a proof.

\begin{remark}
    Under the balance assumption, our algorithms achieve polynomial-time performance for learning intersections of $k=O(1)$ homogeneous halfspaces (see \Cref{corollary:homogeneous-results}). This demonstrates the importance of the balance condition on the training data.
\end{remark}

\subsection{Balanced Intersections of Two General Halfspaces}\label{section:sq-intersections-general-non-degenerate}

We now provide an SQ lower bound for TDS learning balanced (see \Cref{definition:bounded-bias-non-degeneracy}) intersections of two general halfspaces. The lower bound demonstrates that the balance condition cannot always mitigate the obstacles of TDS learning due to hard examples that are trivial for PAC learning. In particular, the hard example here is an intersection of two halfspaces, where one of them is known and the other one is orthogonal to the first and is effectively irrelevant for the intersection under the Gaussian measure. For PAC learning, this implies that the second halfspace can be safely ignored, but for TDS learning, the hidden halfspace is a source of SQ lower bounds as demonstrated below.

\begin{theorem}[SQ Lower bound for TDS Learning Halfspace Intersections]
\label{theorem: TDS learning requires d to the log even balanced}
For $\epsilon>0$, set $d = \eps^{-1/4}$. Then, for all sufficiently small $\epsilon$, the following is true. Let $\mathcal{A}$ be a TDS learning algorithm for $\frac{1}{3}$-balanced intersections of $2$ general halfspaces over $\R^{d}$ w.r.t. $\Gauss_d$, with accuracy parameter $\epsilon$ and success probability at least
$0.95$. Further, suppose that $\mathcal{A}$ obtains at most $d^{\frac{\log1/\epsilon}{\log\log1/\epsilon}}$ samples from the training distribution and accesses the testing distribution via $2^{d^{o(1)}}$ SQ queries of precision $\tol>0$ (the SQ queries
are allowed to depend on the training samples). Then, the tolerance $\tol$ has to be at most $d^{-\Omega(\frac{\log1/\epsilon}{\log\log1/\epsilon})}$.
\end{theorem}

The idea is similar to the one used for the proof of \Cref{theorem: TDS learning requires d to the log}. We once more prove a general reduction of the biased halfspace detection problem to TDS learning. 

The hard instance corresponds once more (as for the proof of \Cref{theorem: TDS learning requires d to the log}) to the detection problem where the unknown distribution is either (1) the standard Gaussian or (2) some distribution $\Dgeneric'$ that assigns non-trivial mass in the negative region of a halfspace $H_1 = \{\x:\vv\cdot \x+\tau\ge 0\}$ for some appropriately large $\tau$. 

The reduction of the hard instance to TDS learning follows closely the proof of \Cref{proposition: tds learning implies biased halfspace detection} (see Appendix \ref{section:detecting-via-tds-appendix}), but we run the TDS algorithm twice, once using training data of the form $(\x,\sign(\vu\cdot \x))$ with $\x\sim \Gauss_d$ and $\vu$ some random vector in $\S^{d-1}$ and another one with training data of the form $(\x,\sign(-\vu\cdot \x))$, $\x\sim\Gauss_d$. 

For each of the executions of the TDS algorithm, the training data are consistent (w.h.p.) with the unknown intersection defined by the halfspaces $H_1=\{\x: \vv\cdot \x + \tau \ge 0\}$ and $H_2 = \{\x: \vu\cdot \x\ge 0\}$ (or $\bar H_2 =\{\x:-\vu\cdot \x\ge 0\}$). If the TDS algorithm rejects, then we have a certificate that the marginal was not the Gaussian. If the TDS algorithm accepts, then we may use one SQ query for the probability that the output function is positive. If $\Dgeneric'$ was the Gaussian, then this probability should be very close to $1/2$. Otherwise, it should be bounded away from $1/2$ for at least one of the executions ($\Dgeneric'$ assigns non-trivial mass in the negative region of $H_1$, so it must assign non-trivial mass to either $H_2\setminus H_1$ or $\bar H_2\setminus H_1$). Hence, the pair of our SQ queries (one for each execution) will indicate the answer to the biased halfspace detection problem.

\begin{remark}
    Note that the lower bound of \Cref{theorem: TDS learning requires d to the log even balanced} holds even for the problem of TDS learning $2$-non-degenerate intersections of two halfspaces (according to \Cref{definition:non-degeneracy}). Under the non-degeneracy assumption, our algorithms achieve improved performance (see \Cref{corollary:general-results}) and, in particular, the lower bound of \Cref{theorem: TDS learning requires d to the log even balanced} is essentially tight ($d^{\tilde\Theta(\log(1/\eps))}$) for TDS learning $\Theta(1)$-non-degenerate, $\poly(\eps)$-balanced intersections of $k = O(1)$ halfspaces.
\end{remark}

\bibliographystyle{alpha}
\bibliography{refs.bib}

\appendix

\crefalias{section}{appendix} 

\newpage
\section{Notation and Basic Definitions}

We let $\R^d$ be the $d$-dimensional Euclidean space. For a distribution $\Dgeneric$ over $\R^d$, we use $\E_\Dgeneric$ (or $\E_{\x\sim \Dgeneric}$) to refer to the expectation over distribution $\Dgeneric$ and for a given (multi)set $\Sunlabelled$, we use $\E_\Sunlabelled$ (or $\E_{\x\sim \Sunlabelled}$) to refer to the expectation over the uniform distribution on $\Sunlabelled$ (i.e., $\E_{\x\sim \Sunlabelled}[g(\x)] = \frac{1}{|\Sunlabelled|}\sum_{\x\in \Sunlabelled}g(\x)$, counting possible duplicates separately).
For $\x\in \R^d$ where $\x = (\x_1,\x_2,\dots,\x_d)$ and for $\mindex \in \N^d$, we denote with $\x^\mindex$ the product $\prod_{i\in[d]}\x_i^{\mindex_i}$. We denote with $\S^{d-1}$ the $d-1$ dimensional sphere on $\R^d$. For any $\vv_1,\vv_2\in\R^d$, we denote with $\vv_1\cdot \vv_2 $ the inner product between $\vv_1$ and $\vv_2$ and we let $\measuredangle (\vv_1,\vv_2)$ be the angle between the two vectors, i.e., the quantity $\theta\in [0,\pi]$ such that $\|\vv_1\|_2\|\vv_2\|_2\cos(\theta) = \vv_1\cdot \vv_2$. Let $\var_{\x}(\vv\cdot \x)$ denotes the variance of random variable $\vv\cdot \x$, for some vector $\vv\in\R^d$. For $\vv\in\R^d,\tau\in \R$, we call a function of the form $\x\mapsto \sign(\vv\cdot \x)$ a homogeneous halfspace and a function of the form $\x\mapsto \sign(\vv\cdot \x + \tau)$ a general halfspace over $\R^d$. An intersection of halfspaces is a function from $\R^d$ to $\cube{}$ of the form $\x\mapsto 2\wedge_{i\in [k]}\ind\{\w^i\cdot \x+\tau^i\ge 0\}-1$, where $\w^i$ are called the normals of the intersection and $\tau^i$ the corresponding thresholds.

\paragraph{Learning Setup.} We focus on the framework of \textbf{testable learning with distribution shift (TDS learning)} defined by \cite{klivans2023testable}. In particular, for a concept class $\C\subseteq\{\R^d\to\cube{}\}$, the learner $\A$ is given $\eps,\delta\in(0,1)$, a set $\Strain$ of labelled examples of the form $(\x,\copt(\x))$, where $\x\sim\Dtrain = \Gauss_d$ and $\copt\in\C$, as well a set $\Stest$ of unlabelled examples from an arbitrary test distribution $\Dtest$ and is asked to output a hypothesis $h:\R^d\to \{\pm 1\}$ with the following guarantees.
\begin{enumerate}[label=\textnormal{(}\alph*\textnormal{)}]
    \item (Soundness.) With probability at least $1-\delta$ over the samples $\Strain,\Stest$ we have: 
    
    If $\A$ accepts, then the output $h$ satisfies $\pr_{\x\sim\Dtest}[\copt(\x)\neq h(\x)] \leq \eps$.
    \item (Completeness.) Whenever $\Dtest = \Gauss_d$, $\A$ accepts w.p. at least $1-\delta$ over $\Strain,\Stest$.
\end{enumerate}
If the learner $\A$ enjoys the above guarantees, then $\A$ is called an $(\eps,\delta)$-TDS learner for $\C$ w.r.t. $\Gauss_d$. Since the probability of success can be amplified through repetition (see \cite[Proposition C.1]{klivans2023testable}), in what follows, we will provide algorithms with constant failure probability.

For our upper bounds, we will make use of a balanced concepts condition, whose importance we justify through appropriate lower bounds (see \Cref{section:sq-general-halfspaces,section:sq-intersections-homogeneous}). In particular, we will assume that the ground truth ($\Dtrain,\copt$) is sufficiently balanced, meaning that positive and negative examples from the training data both have sufficiently large frequency.

\begin{definition}[Balance Condition]\label{definition:bounded-bias-non-degeneracy}
    Let $\Dgeneric$ be a distribution over $\R^d$ and $\concept:\R^d\to \cube{}$. For $\bias\in(0,1/2]$, we say that $\concept$ is $\bias$-balanced with respect to $\Dgeneric$ if 
    \[
        \pr_{\x\sim\Dgeneric}[\concept(\x)=1] \in [\bias,1-\bias]
    \]
    For a concept class $\C\subseteq\{\R^d\to\cube{}\}$, we denote with $\C_\bias$ the $\bias$-balanced version of $\C$, i.e., the subset of $\C$ that contains the elements that are $\bias$-balanced.
\end{definition}

Note that the algorithm can check whether the ground truth is balanced using training data and, therefore, detect possible failure due to imbalance (i.e., the condition is testable). 

\section{Additional Tools}

Our positive results build on the dimension reduction technique of \cite{vempala2010learning} for PAC learning intersections of halfspaces and low-dimensional convex sets through principal component analysis (PCA), which is based on the following Gaussian variance reduction lemma. Note that although the first two parts of the lemma were known (see e.g., \cite{vempala2010learning}), the last part (which gives variance reduction for any vector that has some correlation with a normal) is proven here. In fact, this more general form of the lemma is important even for the results in \cite{vempala2010learning} (although it is missing from the original paper). 

\begin{lemma}[Variance Reduction, variant of Lemma 4.7 in \cite{vempala2010learning}]\label{lemma:variance-reduction}
    Let $\convset\subseteq \R^d$ be an intersection of halfspaces and let $\Gauss_d|_\convset$ be the truncation of the standard Gaussian distribution in $d$ dimensions $\Gauss_d$ to $\convset$. For any $\vu\in\S^{d-1}$, we have $\var_{\x\sim \Gauss_d|_\convset}(\vu\cdot\x) \le 1$. Moreover, if for some $\threshold\in \R$ the halfspace $\{\x: \vu\cdot \x + \threshold \ge 0\}$ is one of the defining halfspaces of the intersection then, we have variance reduction along $\vu$, i.e., $\var_{\x\sim \Gauss_d|_\convset}(\vu\cdot\x) \le 1 - \frac{1}{C}e^{-\frac{1}{2}(\max\{0,\threshold\})^2}$ for a sufficiently large universal constant $C>0$. Furthermore, for any $\eps\in(0,\frac{1}{4})$ and any $\vu'\in\S^{d-1}$ with $\vu\cdot \vu' \ge \eps$, for a sufficiently large constant $C'>0$, if $\bias = \pr_{\Gauss_d}[\x\in \convset]$ we have 
    \[  
        \var_{\x\sim \Gauss_d|_\convset}(\vu'\cdot\x) \le 1 - \bigr(\bias{e^{-{\threshold^2}}/2}\bigr)^{\frac{C'}{\eps^2}}
    \]
\end{lemma}

\begin{proof}
    The first two parts follow from Cafarelli's theorem, see e.g. Theorem 3.1 in \cite{funaki2007dynamic} where one may set the function $\psi$ to be a quadratic function within the interval $(-\threshold,\infty)$ and either $0$ outside it when $\threshold<0$ or a linear function tangent to the graph of $y = x^2$ at the point $x = \threshold$ if $\threshold\ge 0$\footnote{This choice of $\psi$ is due to Raghu Meka \cite{raghu}.}.

    For the last part, we will introduce an artificial halfspace in the direction of $\vu'$ and we will link the variance in the direction of $\vu'$ under the truncation of the Gaussian on the initial intersection to the variance under the new (artificial) truncation. In particular, let $\convset'$ be the set $\convset \cap \{\vu'\cdot \x + \theta \ge 0\}$, where $\theta>0$ is a parameter of our choice. We then have $\var_{\x\sim\Gauss_d|_{\convset'}}(\x) \le 1-\frac{1}{C}\exp(-\theta^2/2)$, by the previous part of the lemma. However, we are interested in the quantity $\var_{\x\sim\Gauss_d|_{\convset}}(\x)$. We have the following
    \begin{align*}
        \var_{\x\sim\Gauss_d|_{\convset}}(\x) = &\E_{\x\sim\Gauss_d|_{\convset}}[(\vu'\cdot \x)^2] - \E_{\x\sim\Gauss_d|_{\convset}}[\vu'\cdot \x]^2 \\
        = &\underbrace{\E_{\x\sim\Gauss_d|_{\convset}}[(\vu'\cdot \x)^2\ind\{\x\in\convset'\}]}_{s_1}+\underbrace{\E_{\x\sim\Gauss_d|_{\convset}}[(\vu'\cdot \x)^2\ind\{\x\not\in\convset'\}]}_{s_2} \\
        &- (\underbrace{\E_{\x\sim\Gauss_d|_{\convset}}[(\vu'\cdot \x)\ind\{\x\in\convset'\}]}_{\mu_1}+\underbrace{\E_{\x\sim\Gauss_d|_{\convset}}[(\vu'\cdot \x)\ind\{\x\not\in\convset'\}]}_{\mu_2})^2
    \end{align*}
    For the first term $s_1$, we have $s_1\le \E_{\x\sim\Gauss_d|_{\convset'}}[(\vu'\cdot \x)^2]$. For the second term $s_2$, we have 
    \begin{align*}
        s_2 &= \frac{\E_{\x\sim\Gauss_d}[(\vu'\cdot \x)^2\ind\{\x\in\convset\setminus\convset'\}]}{\pr_{\x\sim\Gauss_d}[\x\sim\convset]} \\
        &\le \frac{1}{\pr_{\x\sim\Gauss_d}[\x\in\convset]}\cdot \E_{\x\sim\Gauss_d}\Bigr[(\vu'\cdot \x)^2 \ind\Bigr\{\vu'\cdot \x+\theta<0, \vv\cdot \x>\underbrace{\frac{\theta}{\tan \cos^{-1}\eps} - \frac{\threshold}{\sin\cos^{-1}\eps}}_{\gamma}\Bigr\}\Bigr]\,,
    \end{align*}
    where the inequality follows from the fact that for any $\x\in\convset$ we have $\vu\cdot \x+\threshold\ge 0$ and for any $\x\not\in\convset'$ we have $\vu'\cdot \x+\theta<0$, where $\vv = \frac{\vu-(\vu\cdot \vu')\vu'}{\|\vu-(\vu\cdot \vu')\vu'\|_2}$. Hence, by bounding the Gaussian integral of the above inequality (note that $\vu'\perp\vv$), we obtain that for some sufficiently large constant $C'>0$ we have $s_2 \le \frac{1}{\pr_{\x\sim\Gauss_d}[\x\in\convset]}C'\theta^2e^{-\frac{1}{2}\theta^2 - \frac{1}{2}\gamma^2}$. For the term $\mu_1$ we have 
    \begin{align*}
        \mu_1 &= \E_{\Gauss_d|_{\convset'}}[\vu'\cdot \x] \cdot \bigr(1-{\pr_{\Gauss_d|_{\convset}}[\x\not\in\convset']}\bigr) \\
        &= \E_{\Gauss_d|_{\convset'}}[\vu'\cdot \x] \cdot \Bigr(1-\underbrace{\frac{{\pr_{\Gauss_d}[\x\in\convset\setminus\convset']}}{\pr_{\Gauss_d}[\x\in\convset]}}_{\xi}\Bigr)
    \end{align*}
    Therefore, we have that $\mu_1^2 \ge \E_{\Gauss_d|_{\convset'}}[\vu'\cdot \x]^2 - 2\xi \E_{\Gauss_d|_{\convset'}}[\vu'\cdot \x]$. Additionally, we have that $\E_{\Gauss_d|_{\convset'}}[\vu'\cdot \x] = \frac{1}{\pr_{\Gauss_d}[\x\in\convset']}\cdot \E_{\Gauss_d}[(\vu'\cdot \x )\ind\{\x\in\convset'\}] \le \frac{1}{(1-\xi)\pr_{\Gauss_d}[\x\in\convset]}(\E_{\Gauss_d}[(\vu'\cdot \x )^2\ind\{\x\in\convset'\}])^{1/2}$ which implies that $\mu_1^2\ge \E_{\Gauss_d|_{\convset'}}[\vu'\cdot \x]^2 - \frac{2\xi}{(1-\xi)\pr_{\Gauss_d}[\x\in\convset]}$. Note that the quantity $\pr_{\Gauss_d}[\x\in\convset\setminus\convset']$ is bounded by $\pr_{\Gauss_d}[\vu'\cdot \x+\theta<0, \vv\cdot \x>\gamma] \le e^{-\frac{1}{2}\theta^2-\frac{1}{2}\gamma^2}$.

    The term $2\mu_1\mu_2$ can be bounded similarly (observe that $\mu_2\le s_2^{1/2}$). Hence, overall, we have
    \[
        \var_{\x\sim{\Gauss_d|_{\convset}}}(\vu'\cdot \x) \le \var_{\x\sim{\Gauss_d|_{\convset'}}}(\vu'\cdot \x) + \Bigr(\frac{C'\theta^2}{\pr_{\x\sim\Gauss_d}[\x\in\convset]} + \frac{C'}{\pr_{\x\sim\Gauss_d}[\x\in\convset]^2} \Bigr)\cdot e^{-\frac{1}{2}\theta^2-\frac{1}{2}\gamma^2}
    \]
    Recall that $\var_{\x\sim{\Gauss_d|_{\convset'}}}(\vu'\cdot \x) \le1- \frac{1}{C}e^{-\frac{1}{2}\theta^2}$ and hence by picking $\theta = C''\frac{\threshold+{\log^{1/2}({1}/{\bias})}}{\eps}$, where $\bias = \pr_{\Gauss_d}[\x\in\convset]$ and $C''\ge 1$ some sufficiently large constant, we have $\var_{\x\sim{\Gauss_d|_{\convset}}}(\vu'\cdot \x) \le 1-\frac{1}{2C}e^{-\frac{1}{2}\theta^2}$. This concludes the proof of \Cref{lemma:variance-reduction}.
\end{proof}

We will also make use of the following lemma regarding the sample complexity of estimating the expectation and covariance matrix of a log-concave distribution. Note that the truncation of the standard Gaussian on any convex set is log-concave and has variance at most $1$ in every direction.

\begin{lemma}[Mean and Covariance Estimation, see Lemma 4.2 in \cite{vempala2010learning}]\label{lemma:mean-covariance-estimation-log-concave}
    Let $C>0$ be a sufficiently large universal constant, let $\esterror > 0, \delta\in(0,1)$, let $\Dgeneric$ be some log-concave distribution over $\R^d$ such that the variance in every direction is bounded by $1$ and let $\Sunlabelled$ be a set of i.i.d. samples from $\Dgeneric$ of size $|\Sunlabelled|\ge C\cdot \frac{d}{\esterror^2}\log^2(d/\delta)$. Then, with probability at least $1-\delta$, we have 
    \[
        \bigr\|\E_{\x\sim\Sunlabelled}[\x]-\E_{\x\sim\Dgeneric}[\x]\bigr\|_2\le \esterror\text{ and }\bigr\|\var_{\x\sim\Sunlabelled}(\x)-\var_{\x\sim\Dgeneric}(\x)\bigr\|_2\le \esterror
    \]
\end{lemma}

The following lemma is a standard argument that provides a sparse cover of the $k$-dimensional sphere and will be useful in order to exhaustively search in the low-dimensional subspace.

\begin{lemma}[Sparse Cover w.r.t. Angular Distance]\label{lemma:sparse-cover-angles}
    Let $\subspace$ be a linear subspace spanned by the vectors $(\vv^1,\vv^2,\dots,\vv^k)$. For $\eps\in(0,\frac{1}{k})$, let $\subspace_\eps = \{\frac{\vu}{\|\vu\|_2}: \vu = \eps\sum_{i=1}^k j_i \vv^i, j_i\in\Z\cap[-\frac{1}{\eps},\frac{1}{\eps}]\}$. Then, for any $\vv\in \subspace$, there is $\vu\in \subspace_\eps$ such that $\measuredangle(\vv,\vu) \le 6({k\eps})^{1/4}$ and $|\subspace_\eps| \le (\frac{2}{\eps})^k$.
\end{lemma}

\begin{proof}of \Cref{lemma:sparse-cover-angles}, see \cite{vempala2010random}.
    Let $\vv\in\subspace$, which we assume w.l.o.g. to have unit norm (since we only focus on angular distance). We have $\vv=\sum_{i\in[k]}\lambda_i\vv^i$ with $\sum_{i\in[k]}\lambda_i^2 = 1$ and $\lambda_i\in[-1,1]$. For each $i$, there exists $j_i\in\Z\cap[-\frac{1}{\eps},\frac{1}{\eps}]$ such that $|\lambda_i-\eps j_i|\le \eps$. Therefore, if $\vu = \sum_{i\in[k]}\eps j_i \vv^i$, then we have $\vv\cdot \vu \ge 1-k\eps$ and $\|\vu\|_2 \le 1+3\sqrt{k\eps}$, which implies that $\cos(\vu,\vv)\ge \frac{1-k\eps}{1+3\sqrt{k\eps}}\ge 1- 4\sqrt{k\eps}$ and therefore $\measuredangle(\vu,\vv) \le 6(k\eps)^{1/4}$. 
\end{proof}

We will need the following result from \cite{gollakota2023efficient} which provides a tester which ensures that any homogeneous halfspace with normal that is geometrically close to some given vector $\hat\w$ has low disagreement with the halfspace corresponding to $\hat\w$ under the tested marginal.

\begin{lemma}[Tester for Local Halfspace Disagreement, see \cite{gollakota2023efficient}]\label{lemma:homogeneous-disagreement-tester}
    Let $C>0$ be a sufficiently large universal constant. There is a tester that for any $\eps,\delta\in (0,\frac{1}{2})$, any $\hat\w\in\S^{d-1}$ and any (multi)set $\Sunlabelled$ of points in $\R^d$, runs in time $O(d^3+d^2|\Sunlabelled|)$ and satisfies the following.
    \begin{enumerate}[label=\textnormal{(}\alph*\textnormal{)}]
    \item (Soundness.) If the tester accepts, then for any $\w\in\S^{d-1}$, with $\measuredangle(\w,\hat\w) \le \eps$ we have
    \[
        \pr_{\x\sim\Sunlabelled}[\sign(\w\cdot \x)\neq \sign(\hat\w\cdot \x)] \le C\cdot \eps^{\frac{2}{3}}
    \]
    \item (Completeness.) Whenever $\Sunlabelled$ consists of $m\ge C(\frac{1}{\eps^{4/3}}\log(1/\delta)+ d\log^2(d/\delta))$ independent samples from $\Gauss_d$, the tester accepts w.p. at least $1-\delta$.
\end{enumerate}
\end{lemma}

\begin{proof}of \Cref{lemma:homogeneous-disagreement-tester}, combination of Propositions 3.2, 3.3 and 4.5 in \cite{gollakota2023efficient}.
    The tester does the following. 
    \begin{enumerate}
        \item Compute $\pr_{\x\sim\Sunlabelled}[|\hat\w\cdot \x|\le 2\eps^{2/3}]$ and \textbf{reject} if its value is greater than $5\eps^{2/3}$.
        \item Compute the largest eigenvalue of the covariance matrix $\var_{\x\sim\Sunlabelled}(\x)$ and \textbf{reject} if its value is greater than $2$.
        \item Otherwise, \textbf{accept}.
    \end{enumerate}
    \paragraph{Soundness.} If the tester accepts, then we have the following. Suppose that $\w\neq\hat\w$ (otherwise, the proof is trivial). Let $\vv =\frac{\w - (\w\cdot\hat\w)\hat\w}{\|\w - (\w\cdot\hat\w)\hat\w\|_2}$ (so $\vv$ orthogonal to $\hat\w$). Observe that for any $\x$ with $\sign(\w\cdot \x)\neq \sign(\hat\w\cdot \x)$ and $|\hat\w\cdot \x|>2\eps^{2/3}$, it holds that $|\vv\cdot \x|\ge \frac{2\eps^{2/3}}{\tan \eps}$, since we have $|\vv\cdot \x| = \frac{|\w\cdot \x|+|\w\cdot \hat\w|\cdot |\hat\w\cdot \x|}{\|\w-(\w\cdot \hat \w)\hat \w\|_2}$, where $\w\cdot \x\ge 0$, $\w\cdot \hat\w \ge \cos\eps$ and $\|\w-(\w\cdot \hat \w)\hat \w\|_2\le \sin\eps$. Therefore, we obtain the following by additionally using Chebyshev's inequality.
    \begin{align*}
        \pr_{\x\sim\Sunlabelled}[\sign(\w\cdot \x)\neq \sign(\hat\w\cdot \x)] &\le \pr_{\x\sim\Sunlabelled}[|\hat\w\cdot \x|\le 2\eps^{2/3}] + \pr_{\x\sim\Sunlabelled}[|\vv\cdot \x|\ge {2\eps^{2/3}}/{\tan \eps}] \\
        &\le 5\eps^{2/3} + \frac{(\tan\eps)^2 \E_{\x\sim\Sunlabelled}[(\vv\cdot \x)^2]}{4\eps^{4/3}} \\
        &\le 5\eps^{2/3} + 2\eps^{2 -\frac{4}{3}} = 7\eps^{2/3}
    \end{align*}
    \paragraph{Completeness.} For completeness, assume that $\Sunlabelled$ consists of $m$ i.i.d. Gaussian examples. We have that $\E_{\Sunlabelled}[\pr_{\x\sim\Sunlabelled}[|\hat\w\cdot \x|\le 2\eps^{2/3}]] = \pr_{\x\sim\Gauss_d}[|\hat\w\cdot \x|\le 2\eps^{2/3}] \le 4\eps^{2/3}$. By using a standard Hoeffding bound, we have that the first test will accept with probability at least $1-2\delta$ as long as $m\ge \frac{C}{\eps^{4/3}}\log(1/\delta)$ and $C$ is sufficiently large. Moreover, by \Cref{lemma:mean-covariance-estimation-log-concave}, as long as $m\ge C\cdot d\cdot\log^2(d/\delta)$, we have that the largest eigenvalue of $\var_{\x\sim\Sunlabelled}(\x)$ is at most $2$ (since $\|\var_{\x\sim\Gauss_d}(\x)\|_2 = 1$).
\end{proof}

We also prove the following generalization of \Cref{lemma:homogeneous-disagreement-tester} for general halfspaces.

\begin{lemma}[Tester for Local Halfspace Disagreement: General Halfspaces]\label{lemma:general-disagreement-tester}
    Let $C>0$ be a sufficiently large universal constant. There is a tester that for any $\eps,\delta\in (0,\frac{1}{2})$ and $\threshold>0$, any $\hat\w\in\S^{d-1}, \hat\tau\in [-\threshold,\threshold]$ and any (multi)set $\Sunlabelled$ of points in $\R^d$, runs in time $O(d^3+d^2|\Sunlabelled|)$ and
    \begin{enumerate}[label=\textnormal{(}\alph*\textnormal{)}]
    \item (Soundness.) If the tester accepts, then for any $\w\in\S^{d-1}$, $\tau\in \R$, with $\measuredangle(\w,\hat\w) \le \eps$ and $|\tau-\hat\tau|\le \eps$ we have
    \[
        \pr_{\x\sim\Sunlabelled}[\sign(\w\cdot \x+\tau)\neq \sign(\hat\w\cdot \x+\hat\tau)] \le C\eps \threshold + C\eps^{\frac{2}{3}}
    \]
    \item (Completeness.) Whenever $\Sunlabelled$ consists of $m\ge C((\frac{1}{\threshold^2\eps^2}+\frac{1}{\eps^{4/3}})\log(1/\delta)+ d\log^2(d/\delta))$ independent samples from $\Gauss_d$, the tester accepts w.p. at least $1-\delta$.
\end{enumerate}
\end{lemma}

\begin{proof}of \Cref{lemma:general-disagreement-tester}.
    The tester does the following for $\gamma = 10(\eps\threshold + \eps^{2/3})$. 
    \begin{enumerate}
        \item Compute $\pr_{\x\sim\Sunlabelled}[|\hat\w\cdot \x+\hat\tau|\le \gamma]$ and \textbf{reject} if its value is greater than $5\gamma$.
        \item Compute the largest eigenvalue of the covariance matrix $\var_{\x\sim\Sunlabelled}(\x)$ and \textbf{reject} if its value is greater than $2$.
        \item Otherwise, \textbf{accept}.
    \end{enumerate}
    \paragraph{Soundness.} If the tester accepts, then we have the following. Suppose that $\w\neq\hat\w$ (otherwise, the proof is trivial). Let $\vv =\frac{\w - (\w\cdot\hat\w)\hat\w}{\|\w - (\w\cdot\hat\w)\hat\w\|_2}$ (so $\vv$ orthogonal to $\hat\w$). Observe that for any $\x$ with $\sign(\w\cdot \x+\tau)\neq \sign(\hat\w\cdot \x+\hat\tau)$ and $|\hat\w\cdot \x+\hat\tau|>\gamma$, we have the following.
    \begin{align*}
        |\vv\cdot \x| & = \frac{|\w\cdot \x - (\w\cdot\hat\w)\hat\w\cdot \x|}{\|\w - (\w\cdot\hat\w)\hat\w\|_2} \\
        & = \frac{|\w\cdot \x+\tau - \tau +\hat\tau(\w\cdot\hat\w) - (\w\cdot\hat\w)(\hat\w\cdot \x+\hat\tau)|}{\|\w - (\w\cdot\hat\w)\hat\w\|_2} \\
        &\ge \frac{|\w\cdot \x+\tau|+ |(\w\cdot\hat\w)(\hat\w\cdot \x+\hat\tau)| - |\tau -\hat\tau(\w\cdot\hat\w)| }{\|\w - (\w\cdot\hat\w)\hat\w\|_2}\,,
    \end{align*}
    where for the first equality we add and subtract the terms $\tau$ and $\hat\tau(\w\cdot \hat\w)$ and for the inequality we use the fact that the signs of the halfspaces are opposite. Moreover, since we have $|\w\cdot \x+\tau|\ge 0$, $|\w\cdot \hat\w| \ge \cos\eps$, $|\hat\w\cdot \x+\hat\tau|>\gamma$ and $|\hat\tau-\tau|\le \eps$, $|\hat\tau|\le \threshold$, we obtain the following.
    \[
        |\vv\cdot \x| \ge \frac{\gamma \cos \eps - \threshold |1-\cos\eps| - \eps}{\sin\eps} \ge \frac{\gamma \cos\eps - \eps(\threshold+1)}{\sin \eps} \ge \frac{\gamma}{\tan \eps} - (\threshold+1) =: \beta
    \]
    Therefore, we obtain the following by additionally using Chebyshev's inequality.
    \begin{align*}
        \pr_{\x\sim\Sunlabelled}[\sign(\w\cdot \x+\tau)\neq \sign(\hat\w\cdot \x+\hat\tau)] &\le \pr_{\x\sim\Sunlabelled}[|\hat\w\cdot \x+\hat\tau|\le \gamma] + \pr_{\x\sim\Sunlabelled}[|\vv\cdot \x|\ge \beta] \\
        &\le 3\gamma + \frac{\E_{\x\sim\Sunlabelled}[(\vv\cdot \x)^2]}{\beta^2} \\
        &\le 3\gamma + \frac{2}{\beta^2} \le C'\gamma\,,
    \end{align*}
    for a sufficiently large constant $C'>0$, due to the choice of $\gamma$.
    \paragraph{Completeness.} For completeness, assume that $\Sunlabelled$ consists of $m$ i.i.d. Gaussian examples. We have that $\E_{\Sunlabelled}[\pr_{\x\sim\Sunlabelled}[|\hat\w\cdot \x+\hat\tau|\le \gamma]] = \pr_{\x\sim\Gauss_d}[|\hat\w\cdot \x+\hat\tau|\le \gamma] \le 2\gamma$. By using a standard Hoeffding bound, we have that the first test will accept with probability at least $1-2\delta$ as long as $m\ge \frac{C}{\gamma^2}\log(1/\delta)$ and $C$ is sufficiently large. Moreover, by \Cref{lemma:mean-covariance-estimation-log-concave}, as long as $m\ge C\cdot d\cdot\log^2(d/\delta)$, we have that the largest eigenvalue of $\var_{\x\sim\Sunlabelled}(\x)$ is at most $2$ (since $\|\var_{\x\sim\Gauss_d}(\x)\|_2 = 1$).
\end{proof}

Finally, we state the following result from \cite{klivans2023testable}, which demonstrates that any high bias halfspace behaves as a constant function with respect to any distribution that matches sufficiently many moments up to sufficiently small accuracy with the Gaussian distribution.

\begin{lemma}[Concentration via Moment Matching, see Lemma 5.6 in \cite{klivans2023testable}]\label{lemma:moment-matching-concentration}
    Let $\eps>0$. Suppose that $\Sunlabelled$ is a set of points in $\R^d$ such that the empirical moments of bounded degree the uniform distribution over $\Sunlabelled$ approximately match the corresponding moments of the standard Gaussian, i.e., $|\E_{\x\sim\Sunlabelled}[\x^\mindex] - \E_{\x\sim\Gauss_d}[\x^\mindex]| \le d^{-\log(1/\eps)}$ for any $\mindex\in\N^d$ s.t. $\|\mindex\|_1 \le \log(1/\eps)$. Then, for any $\w\in\S^{d-1}$ and $\tau\in\R$, with $|\tau| \ge 3\sqrt{\log (1/\eps)}$ we have that 
    \[
        \pr_{\x\sim\Stest}[\sign(\w\cdot \x+\tau) \neq \sign(\tau)] \le \eps
    \]
\end{lemma}

\section{Approximate Subspace Retrieval}\label{section:appendix-pca}

In this section we provide a number of subspace retrieval lemmas, originally from \cite{vempala2010learning} (see \Cref{section:appendix-recovery-general,section:appendix-recovery-non-degenerate}) and \cite{vempala2010random} (see \Cref{section:appendix-recovery-polar}). For the subspace retrieval lemma from \cite{vempala2010learning}, we provide a detailed proof here, but we incur an exponential dependence on $1/\eps^2$. In fact, it is not clear whether our analysis can be improved, since the original proof by \cite{vempala2010learning} has a gap
and, unless a stronger version of \Cref{lemma:variance-reduction} is proven, the complexity of the algorithm in \cite{vempala2010learning} should involve a term of $2^{\poly(k/\eps)}$ as well. To circumvent this obstacle, we also provide a fully polynomial upper bound, under some non-degeneracy assumption (see \Cref{section:appendix-recovery-non-degenerate}).

\subsection{Subspace Retrieval through PCA for Balanced Intersections}\label{section:appendix-recovery-general}

In this section, we will present a proof of \Cref{lemma:subspace-retrieval-general}, which was originally proven by \cite{vempala2010learning}. The idea of the proof is not novel, but we provide a detailed and complete version of it for concreteness. We restate the lemma here for convenience.

\begin{lemma}[Subspace Retrieval, modification from \cite{vempala2010learning}]\label{lemma:subspace-retrieval-general}
    Let $C\ge 1$ be a sufficiently large universal constant. Let $\C$ be the class of intersections of $k$ general halfspaces on $\R^d$, $\eps\in(0,1)$, $\threshold>0$ and $\bias\in(0,1/2]$. Let $\Slabelled$ be a set of at least $dk^4(1/\bias)^{C/\eps^2}2^{CT^2/\eps^2}\log^2(d/\delta)$ labelled examples of the form $(\x,\copt(\x))$, where $\x\sim\Gauss_d$ and $\copt\in\C_\bias$ is an $\bias$-unbiased intersection which is defined by the normal vectors $(\w^1,\dots,\w^k)$ and the corresponding thresholds $(\tau^1,\dots,\tau^k)$. Then, with probability at least $1-\delta$, the subspace $\subspace$ spanned by the $k$-smallest variance orthogonal components of the positive examples $\Slabelled^+=\{\x: (\x,1)\in \Slabelled\}$ approximately includes all of the normal vectors corresponding to bounded thresholds, i.e., for any $i\in[k]$ if $\tau^i\le\threshold$, then $\|\proj_\subspace \w^i\|_2 \ge 1-\eps$.
\end{lemma}

\begin{algorithm2e}
\caption{Subspace Retrieval through PCA}\label{algorithm:pca}
\KwIn{Labelled set $\Strain$, parameter $k$}
\KwOut{Orthonormal basis $(\vv^1,\dots,\vv^k)$}
Let $\Strain^+$ be the subset of $\Strain$ corresponding to positive examples. \\
Run Principal Component Analysis on $\Strain^+ = \{\x: (\x,1)\in\Strain\}$ and let $\vv^1,\dots,\vv^k$ be the $k$ smallest-variance orthogonal components (i.e., the right singular vectors corresponding to the $k$ smallest singular values of the $(|\Strain^+|\times d)$-dimensional sample matrix). \\
\textbf{Output} $(\vv^1,\dots,\vv^k)$ and terminate.
\end{algorithm2e}

For the proof, we will use the following strong theorem which ensures that the subspace retrieved by PCA on the empirical distribution will be geometrically close to the true corresponding subspace, as long as there is a spectral gap in the covariance matrix of the true distribution.

\begin{proposition}[Davis-Kahan, modification of Theorem 2 in \cite{yu2015useful}]\label{proposition:davis-kahan}
    Let $\covar\in\R^{d\times d}$ and $\hat\covar\in\R^{d\times d}$ be symmetric matrices such that for some $k\in[d]$, the gap between the $k$-th smallest eigenvalue of $\covar$ and the $(k+1)$-th smallest eigenvalue of $\covar$ is positive, i.e., $\lambda_{k+1}-\lambda_k > 0$. Let $\vv^1,\dots,\vv^k$ be the eigenvectors of $\covar$ corresponding to the $k$ smallest eigenvalues and, similarly, $\vu^{1},\dots,\vu^{k}$ the $k$ smallest eigenvectors of $\hat\covar$. Then we have that
    \[
        \sum_{i\in[k]}\sin^2(\measuredangle(\vv^i,\vu^i)) \le \frac{4k \|\covar-\hat\covar\|_2^2}{(\lambda_{k+1}-\lambda_k)^2}
    \]
\end{proposition}

Let $\truesubspace$ be the span of $(\w^1,\dots,\w^k)$ and note that every direction orthogonal to $\truesubspace$ has variance $1$ under $\Gauss_d|_{\convset}$. Let $\gamma = (1/\bias)^{C/\eps^2}2^{CT^2/\eps^2}$ and let $\truesubspace_\gamma$ be the subspace of $\truesubspace$ such that for every direction $\vu$ orthogonal to $\truesubspace_\gamma$, we have $\var_{\x\sim\Gauss_d|_{\convset}}(\vu\cdot \x) >1-\gamma$ and $\truesubspace_\gamma$ is spanned by an orthonormal basis $(\z^1,\dots,\z^\ell)$ with $\var_{\x\sim\Gauss_d|_{\convset}}(\z^i\cdot \x)\le 1-\gamma$. In other words, $\truesubspace_\gamma$ is the span of the eigenvectors of the covariance matrix $\covar$ of $\Gauss_d|_{\convset}$ whose corresponding eigenvalues are at most $1-\gamma$. Note that since $\dim(\truesubspace)\le k$ and $\truesubspace_\gamma \subseteq \truesubspace$, we have $\ell\le k$. Let $0\le \lambda_1\le \dots\le \lambda_\ell\le 1-\gamma < \lambda_{\ell+1} \le \dots\le \lambda_k \le 1 = \lambda_{k+1}$ be the $k+1$ smallest eigenvalues of the covariance matrix of $\Gauss_d|_{\convset}$. Since there is a $\gamma$ gap between $\lambda_\ell$ and $\lambda_{k+1}$, there is some $j\in[\ell,k]$ such that $\lambda_{j+1}-\lambda_j > \frac{\gamma}{k}$.

Let $\subspace$ be the subspace corresponding to the $k$ smallest eigenvectors of the empirical covariance matrix $\hat\covar$ of the set of positive examples $\Slabelled^+$. Since $|\Slabelled| \ge \frac{1}{\bias^2}\log(1/\delta)$, due to a Hoeffding bound, we have that with probability at least $1-\delta/10$, $|\Slabelled^+| \ge \frac{\eta}{2}|\Slabelled| \ge dk^4(1/\bias)^{C/\eps^2}2^{CT^2/\eps^2}\log^2(d/\delta)$. We can therefore apply \Cref{lemma:mean-covariance-estimation-log-concave} to $\Gauss_d|_{\convset}$ (which is log-concave) to obtain that $\|\covar-\hat\covar\|_2 \le \frac{\gamma\eps}{2C'k^{2}}$. Let $\subspace_\ell$ be the subspace of $\subspace$ corresponding to the $\ell$ smallest eigenvalues of $\hat\covar$, and let $(\vv^1,\dots,\vv^\ell)$ be the corresponding eigenvectors. By \Cref{proposition:davis-kahan}, we have that 
\begin{equation}\label{equation:sin-theta-theorem-result}
    \sum_{i\in[\ell]}\sin^2(\measuredangle(\vv^i,\z^i)) \le \eps/(C'\sqrt{k})
\end{equation}
    
Let $i\in[k]$ such that $\tau^i\le \threshold$. We analyze $\w^i$ in two orthogonal components, $\w$ and $\w'$, where $\w$ is the normalized projection of $\w^i$ on $\truesubspace_\gamma$ and $\w'$ is therefore orthogonal to $\truesubspace_\gamma$. Since $\w'$ is orthogonal to $\truesubspace_\gamma$, by the definition of $\truesubspace_\gamma$, we have $\var_{\x\sim\Gauss_d}(\w'\cdot \x) > 1-\gamma$. By \Cref{lemma:subspace-retrieval-general}, this implies that $\w^i\cdot \w' < C''\frac{1+\threshold+\log^{1/2}(1/\bias)}{\log^{1/2}(1/\gamma)}$. Therefore, $\measuredangle(\w^i,\w) \le 2C''\frac{1+\threshold+\log^{1/2}(1/\bias)}{\log^{1/2}(1/\gamma)}$. Moreover, by \Cref{equation:sin-theta-theorem-result}, we have that $\measuredangle(\w,\proj_{\subspace_\ell}\w) \le \eps/10$. Since $2C''\frac{1+\threshold+\log^{1/2}(1/\bias)}{\log^{1/2}(1/\gamma)} \le \eps/10$ by the choice of $\gamma$, we obtain the desired result.

\subsection{Subspace Retrieval through PCA under a Non-Degeneracy Assumption}\label{section:appendix-recovery-non-degenerate}

In the previous subsection we provided a detailed proof of the subspace retrieval lemma which was originally proven in \cite{vempala2010learning}, incurring, however, an exponential dependence on $1/\eps^2$. Here, we define a technical assumption on the concept class considered which is sufficient to provide a fully polynomial result for subspace retrieval. Despite its technicality, the non-degeneracy condition is satisfied by the constructions we use for our lower bounds, which implies that under the non-degeneracy condition, our upper and lower bounds are directly comparable (and tight in some regimes).

\begin{definition}[Non-Degeneracy Condition]\label{definition:non-degeneracy}
    Let $\convset$ be an intersection of halfspaces in $\R^d$ and $\Gauss_d|_{\convset}$ be the truncation of the standard Gaussian to $\convset$. For $\nondegen\ge 1$, we say that $\convset$ is $\nondegen$-non-degenerate if the following is true. For every subspace $\truesubspace$ spanned by some of the normals of $\convset$ and for every vector $\w\in\S^{d-1}$ that is a normal to $\convset$ with non-zero projection $\w'\in\R^{d}\setminus\{0\}$ onto the subspace orthogonal to $\truesubspace$ we have
    \[
         \var_{\x\sim\Gauss_d}(\hat{\vec \w}'\cdot \x)-\var_{\x\sim\Gauss_d|_{\convset}}(\hat{\vec \w}'\cdot \x) \ge \bigr( \var_{\x\sim\Gauss_d}(\vec \w\cdot \x)-\var_{\x\sim\Gauss_d|_{\convset}}(\vec \w\cdot \x) \bigr)^\nondegen\,, \text{ where }\hat{\w}' = \w'/\|\w'\|_2
    \]
    For any class $\C$ of halfspace intersections on $\R^d$, we denote with $\C^\nondegen$ the $\nondegen$-non-degenerate version of $\C$, i.e., the subset of $\C$ that contains the elements that are $\nondegen$-non-degenerate.
\end{definition}

The condition defined above states that each normal $\w$ of the intersection has either zero or non-trivial relative influence on subspaces orthogonal to the span $\truesubspace'$ of any subset of the normals. The influence is measured in terms of the variance reduction along the residual direction $\w-\proj_{\truesubspace'}(\w)$. In particular, in light of the third part of \Cref{lemma:variance-reduction}, for intersections of two halfspaces, the non-degeneracy condition is satisfied whenever the two halfspaces of the intersection have normals either pointing to the exact same direction or have sufficiently large angular distance (but nothing in between). This enables one to circumvent the need for a strong quantitative statement relating (1) the angle between some vector $\vu$ and a normal with (2) the variance reduction along $\vu$, which is the source of the exponential dependence of $2^{1/\eps^2}$. With an analysis similar to the one of \Cref{section:appendix-recovery-general}, we obtain the following subspace retrieval result.

\begin{lemma}[Subspace Retrieval under Non-Degeneracy, see \cite{vempala2010learning}]\label{lemma:subspace-retrieval-non-degenerate-appendix}
    Let $C\ge 1$ be a sufficiently large universal constant. Let $\C$ be the class of intersections of $k$ general halfspaces on $\R^d$, $\eps\in(0,1)$, $\threshold\ge 0$ and $\nondegen\ge 1, \bias\in(0,1/2]$. Let $\Slabelled$ be a set of at least $\frac{Cdk^4}{\eps^2\bias^2}e^{\nondegen T^2}\log^2(d/\delta)$ labelled examples of the form $(\x,\copt(\x))$, where $\x\sim\Gauss_d$ and $\copt\in\C_\bias^\nondegen$ is an $\bias$-unbiased and $\nondegen$-non-degenerate intersection which is defined by the normal vectors $(\w^1,\dots,\w^k)$ and the corresponding thresholds $(\tau^1,\dots,\tau^k)$. Then, with probability at least $1-\delta$, the subspace $\subspace$ spanned by the $k$-smallest variance orthogonal components of the positive examples $\Slabelled^+=\{\x: (\x,1)\in \Slabelled\}$ approximately includes all of the normal vectors corresponding to bounded thresholds, i.e., for any $i\in[k]$ if $\tau^i\le\threshold$, then $\|\proj_\subspace \w^i\|_2 \ge 1-\eps$.
\end{lemma}

\subsection{Subspace Retrieval through Polar Planes algorithm}\label{section:appendix-recovery-polar}

We now present the following lemma from \cite{vempala2010random} which provides another algorithm for approximately retrieving the relevant subspace for homogeneous intersections whose runtime is not exponential in $1/\eps$, even without making a non-degeneracy assumption. The lemma follows from combining Theorem 4 and Lemma 3 from \cite{vempala2010random}.

\begin{lemma}[Subspace Retrieval through Polar Planes, from \cite{vempala2010random}]\label{lemma:subspace-retrieval-homogeneous-polar-planes}
    Consider $\C$ to be the class of intersections of $k$ homogeneous halfspaces on $\R^d$, $\eps\in(0,1)$ and $\bias\in(0,1/2]$. Let $\Slabelled$ be a set of at least $m = d(\frac{k}{\eps \bias})^{O(k)} \log(1/\delta)$ labelled examples of the form $(\x,\copt(\x))$, where $\x\sim\Gauss_d$ and $\copt\in\C_\bias$ is an $\bias$-balanced intersection which is defined by the normal vectors $(\w^1,\dots,\w^k)$. There is an algorithm (Polar Planes from \cite{vempala2010random}) that on input $\Slabelled$, returns, w.p. at least $1-\delta$, an orthonormal basis for a subspace $\subspace$ of dimension $k$ that approximately includes all of the normal vectors, i.e., for any $i\in[k]$, we have $\|\proj_\subspace \w^i\|_2 \ge 1-\eps$, in time $(\frac{dk}{\eps\bias})^{O(k)}$.
\end{lemma}

\section{TDS Learning Intersections of Halfspaces}\label{section:appendix-tds-upper-bounds}

We now provide full proofs for all of our upper bounds, assuming the balanced concepts condition (\Cref{definition:bounded-bias-non-degeneracy}), both with and without assuming the non-degeneracy condition (\Cref{definition:non-degeneracy}).

\subsection{Homogeneous Halfspace Intersections}\label{section:appendix-tds-homogeneous}

We prove our result on learning intersections of homogeneous halfspaces, which we restate here for convenience.

\begin{theorem}[TDS Learning Intersections of Homogeneous Halfspaces]\label{theorem:tds-homogeneous-intersections-appendix}
    Let $\C$ be a class whose elements are intersections of $k$ homogeneous halfspaces on $\R^d$, $\eps\in (0,1)$ and $C\ge 1$ a sufficiently large constant.
    \begin{itemize}
        \item Assume that there is an algorithm $\A$ that upon receiving at least $\mrec$ examples from a training distribution of the form $(\x,\copt(\x))$, where $\x\sim\Gauss_d$ and  $\copt\in\C$, outputs, with probability at least $0.99$ an orthonormal basis for a subspace $\subspace$ such that for any normal $\w$ of $\copt$ we have $\|\proj_\subspace \w\|_2 \ge 1 - (\frac{k}{C\eps})^3$.
    \end{itemize}   
    
    Then, there is an algorithm (\Cref{algorithm:tds-homogeneous-intersections-appendix}) that $(\eps,\delta=0.02)$-TDS learns the class $\C$, using $\mrec+ \tilde{O}(\frac{dk^2}{\eps^2})$ labelled training examples and $\tilde{O}(\frac{dk^2}{\eps^2})$ unlabelled test examples, calls $\A$ once and uses additional time $\tilde{O}(\frac{d^3k^2}{\eps^2})+d(k/\eps)^{O(k^2)}$.
\end{theorem}

\begin{algorithm2e}
\caption{Proper TDS Learner for Homogeneous Halfspace Intersections}\label{algorithm:tds-homogeneous-intersections-appendix}
\KwIn{Labelled set $\Strain$, unlabelled set $\Stest$, parameter $\epsilon$}
Set $\eps' = \frac{\eps^{3/2}}{Ck^{3/2}}$ and $\eps'' = \frac{\eps^6}{Ck^{7}}$ for some sufficiently large universal constant $C\ge 1$.\\
Run algorithm $\A$ on the set $\Strain$ and let $(\vv^1,\dots,\vv^k)$ be its output. \\
Let $\subspace$ be the subspace spanned by $(\vv^1,\dots,\vv^k)$ and consider the following sparse cover of $\subspace$: 
$\subspace_{\eps''} = \{\frac{\vu}{\|\vu\|_2}: \vu = \eps''\sum_{i=1}^k j_i \vv^i, j_i\in\Z\cap[-\frac{1}{\eps''},\frac{1}{\eps''}],\|\vu\|_2\neq 0\}$ \\
\textbf{Reject} and terminate if $\|\var_{\x\sim\Sunlabelled}(\x)\|_2 \ge 2$. \\
\For{$\vu\in \subspace_{\eps''}$}{
    \textbf{Reject} and terminate if $\pr_{\x\sim\Sunlabelled}[|\vu\cdot \x|\le 2\eps'^{2/3}] > 5\eps'^{2/3}$. \\
} 
Let $\F$ contain the concepts $\concept:\R^d\to\{\pm 1\}$ of the form $\concept(\x)=2\bigwedge_{i=1}^k \ind\{\vu^i\cdot \x \ge 0\} - 1$, where $\vu^1,\dots,\vu^k\in \subspace_{\eps''}$ and $\pr_{(\x,y)\sim \Strain}[y\neq\concept(\x)] \le \eps/5$. \\
\textbf{Reject} and terminate if $\max_{\concept_1,\concept_2\in \F} \pr_{\x\sim \Stest}[\concept_1(\x)\neq \concept_2(\x)] > \eps/2$.\\
\textbf{Otherwise,} output $\hat{\concept}:\R^d\to\cube{}$ for some $\hat\concept\in \F$.
\end{algorithm2e}

\begin{proof}of \Cref{theorem:tds-homogeneous-intersections}. Let $\Strain$ be a set of $\mtrain$ samples from the training distribution, i.e., of the form $(\x,\copt(\x))$, where $\x\sim \Dtrain=\Gauss_d$ and let $\Stest$ be a set of $\mtest$ samples from the test distribution $\Dtest$. Let $C>0$ be a sufficiently large universal constant. Let $\copt:\R^d\to\cube{}$ denote the ground truth, i.e., the intersection of $k$ homogeneous halfspaces
\[
    \copt(\x) = 2\wedge_{i\in[k]}\ind\{\vec \w^i\cdot \x \ge 0\}-1 \,,\text{ for some }\w^1,\dots,\w^k\in \S^{d-1}
\]
In the following, we will say that an event holds with high probability if it holds with probability sufficiently close to $1$ so that union bounding over all the bad events gives a probability of failure of at most $0.01$. This is possible by choosing $C$ to be a sufficiently large constant.

\paragraph{Soundness.} To prove soundness, suppose that the tests have accepted. We first use the approach of \cite{vempala2010learning} to show that using training data, we can retrieve a subspace that is geometrically close to the normal subspace of the ground truth. Let $C',C''$ be sufficiently large universal constants.

In particular, the guarantee for algorithm $\A$ implies that the retrieved subspace $\subspace$ has the property that for any $i\in[k]$ we have $\|\proj_\subspace \w^i\|_2 \ge 1-(\frac{\eps}{C'k})^3$ with high probability, as long as $\mtrain \ge \mrec$. Let $\w^i_\subspace = \frac{\proj_\subspace \w^i}{\|\proj_\subspace \w^i\|_2}$. Then, we have $\measuredangle(\w^i,\w^i_\subspace) \le \frac{4\eps^{3/2}}{C'k^{3/2}}$. Due to \Cref{lemma:sparse-cover-angles}, there is a vector $\vu^i\in\subspace_{\eps''}$ with $\measuredangle(\vu^i,\w^i_\subspace) \le \frac{\eps^{3/2}}{C'k^{3/2}}$, whenever $\eps'' \le \frac{\eps^6}{6^4C'k^{7}}$, in which case, $|\subspace_{\eps''}|\le (\frac{2\cdot 6^4 C' k^{7}}{\eps^6})^k$. Therefore, for any $i\in[k]$ we have some vector $\vu^i$ in the cover $\subspace_{\eps''}$ that is close to the normal $\w^i$, i.e., $\measuredangle(\vu^i,\w^i) \le (\frac{5\eps}{C'k})^{3/2}$. 

Consider now the hypothesis $\concept(\x) = 2\wedge_{i\in[k]}\ind\{\vu^i\cdot \x\ge 0\} - 1$. If suffices to show that $\concept$ belongs in the set $\F$ of candidate concepts and that $\concept$ has small test error $\pr_{\x\sim\Stest}[\concept(\x)\neq \copt(\x)] \le \eps/4$, because then for any other candidate concept $\concept'\in\F$, we know that it disagrees with $\concept$ only on a small fraction of test points and, hence, we will have $\pr_{\x\sim\Stest}[\concept'(\x)\neq \copt(\x)]\le 3\eps/4$. By standard VC dimension arguments, this would imply that, whenever $\mtest\ge C\frac{dk\log k}{\eps^2}$, with high probability, the test error of any element of $\F$ satisfies $\pr_{\x\sim\Dtest}[\concept'(\x)\neq \copt(\x)]\le \eps$.

We appeal to the tester for local halfspace disagreement of \Cref{lemma:homogeneous-disagreement-tester} in order to demonstrate that $\pr_{\x\sim\Stest}[\concept(\x)\neq \copt(\x)] \le \eps/4$. In particular, we have that
\begin{align*}
    \pr_{\x\sim\Stest}[\concept(\x)\neq \copt(\x)] &\le k \pr_{\x\sim\Stest}[\sign(\vu^i\cdot \x)\neq \sign(\w^i\cdot \x)] \\
    &\le C'' k(\measuredangle(\vu^i,\w^i))^{2/3} \le \eps/4
\end{align*}

Finally, we show that the hypothesis $\concept$ lies within $\F$. In particular, $\pr_{\x\sim\Dtrain}[\concept(\x)\neq \copt(\x)] \le k\pr_{\x\sim\Gauss_d}[\sign(\vu^i\cdot \x)\neq\sign(\w^i\cdot \x)] = O( k\measuredangle(\vu^i,\w^i))$, which is bounded by $\eps/10$ by choosing the constant $C'$ appropriately. By standard VC dimension arguments, we therefore have that $\pr_{\x\sim\Strain}[\concept(\x)\neq \copt(\x)] \le \eps/5$ as long as $\mtrain \ge \frac{Cdk\log k}{\eps^2}$.

\paragraph{Completeness.} To prove completeness, suppose that $\Dtest = \Gauss_d$. Since $\subspace_{\eps''},\F$ do not depend on $\Stest$, we can use Hoeffding bounds to bound the probability of rejection, as well as union bounds over $\F\times \F$ accordingly. In particular, the tester of \Cref{lemma:homogeneous-disagreement-tester} will accept with high probability as long as $\mtest \ge C\frac{1}{\eps'^{4/3}} + Cd\log^2d = O(\frac{k^2}{\eps^2} + d\log^2d)$ and the tester of the disagreement probabilities of pairs in $\F$ will accept (due to standard Hoeffding and union bounds) with high probability whenever $\mtest \ge C \frac{1}{\eps^2} \log|\F| = O(\frac{k^2}{\eps^2} \log(\frac{k}{\eps}))$ (since $|\F| = (k/\eps)^{O(k^2)}$ as we need to choose $k$ normals from $\subspace_{\eps''}$).
\end{proof}

By combining \Cref{theorem:tds-homogeneous-intersections-appendix} with \Cref{lemma:subspace-retrieval-general,lemma:subspace-retrieval-non-degenerate-appendix,lemma:subspace-retrieval-homogeneous-polar-planes} we obtain the following bounds for TDS learning homogeneous halfspace intersections.

\begin{corollary}[TDS Learning Bounds for Homogeneous Halfspace Intersections]\label{corollary:homogeneous-results}
    Let $\bias\in(0,\frac{1}{2})$, $\eps>0$, $\nondegen\ge 1$ and let $\C$ be the class of intersections of $k$ homogeneous halfspaces on $\R^d$.
    \begin{enumerate}[label=\textnormal{(}\alph*\textnormal{)}]
        \item There is an $(\eps,\delta=0.02)$-TDS learner for the class $\C_\bias$ of $\bias$-balanced intersections that uses $\tilde{O}(d)(\frac{k}{\eps\bias})^{O(\frac{k^6}{\eps^6})}$ labelled training examples, $\tilde{O}(\frac{dk^2}{\eps^2})$ unlabelled test examples and runs in time $\tilde{O}(d^3)(\frac{k}{\eps\bias})^{O(\frac{k^6}{\eps^6})}$.
        \item There is an $(\eps,\delta=0.02)$-TDS learner for the class $\C_\bias^\nondegen$ of $\bias$-balanced and $\nondegen$-non-degenerate intersections that uses $\tilde{O}(d)\cdot \frac{1}{\bias^2}\cdot (\frac{k}{\eps})^{O(\nondegen)}$ labelled training examples, $\tilde{O}(\frac{dk^2}{\eps^2})$ unlabelled test examples and runs in time $\tilde{O}(d^3)\cdot \frac{1}{\bias^2}\cdot (\frac{k}{\eps})^{O(\nondegen)} + d(k/\eps)^{O(k^2)}$.
        \item There is an $(\eps,\delta=0.02)$-TDS learner for the class $\C_\bias$ of $\bias$-balanced intersections that uses $\tilde{O}(d)(\frac{k}{\eps\bias})^{O(k)}$ labelled training examples, $\tilde{O}(\frac{dk^2}{\eps^2})$ unlabelled test examples and with time complexity $(\frac{dk}{\eps\bias})^{O(k)}+d(k/\eps)^{O(k^2)}$.
    \end{enumerate}
\end{corollary}

\subsection{General Halfspace Intersections}\label{section:appendix-tds-general}

We now prove our positive results on learning intersections of general halfspaces.

\begin{theorem}[TDS Learning Intersections of General Halfspaces]\label{theorem:tds-general-intersections-appendix}
    Let $\C$ be a class whose elements are intersections of $k$ general halfspaces on $\R^d$, $\eps,\threshold\in (0,1)$ and $C\ge 1$ a sufficiently large constant.
    \begin{itemize}
        \item Assume that there is an algorithm $\A$ that upon receiving at least $\mrec$ examples of the form $(\x,\copt(\x))$, where $\x\sim\Gauss_d$ and  $\copt\in\C$, outputs, with probability at least $0.99$ an orthonormal basis for a subspace $\subspace$ such that for any normal $\w\in\S^{d-1}$ that corresponds to some halfspace $\{\x:\w\cdot \x + \tau \ge 0\}$ of $\copt$ with threshold $\tau \le \threshold$ we have $\|\proj_\subspace \w\|_2 \ge 1 - (\frac{k}{C\eps})^3$.
    \end{itemize} 
    
    Then, there is an algorithm (\Cref{algorithm:tds-general-intersections-appendix}) that $(\eps,\delta=0.02)$-TDS learns the class $\C$, using $\mrec+ \tilde{O}(\frac{dk^2}{\eps^2})$ labelled training examples and $d^{O(\log(k/\eps))}$ unlabelled test examples, calls $\A$ once and uses additional time $d^{O(\log(k/\eps))}(k/\eps)^{O(k^2)}$.
\end{theorem}

\begin{algorithm2e}
\caption{Proper TDS Learner for General Halfspace Intersections}\label{algorithm:tds-general-intersections-appendix}
\KwIn{Labelled set $\Strain$, unlabelled set $\Stest$, parameter $\epsilon$}
Set $\threshold = 3\log^{1/2}(\frac{10k}{\eps})$, $\mdegree \ge \log(10k/\eps)$, $\mslack = d^{-\mdegree}$, $\eps' = \frac{\eps^{3/2}}{Ck^{3/2}}$ and $\eps'' = \frac{\eps^6}{Ck^{3/2}}$, where $C\ge1$ is a sufficiently large constant.\\
\textbf{Reject} and terminate if for some $\mindex\in\N^d$ with $\|\mindex\|_1\le \mdegree$ it holds $|\E_{\x\sim\Stest}[\x^\mindex]-\E_{\x\sim\Gauss}[\x^\mindex]| > \mslack$ \\
Run algorithm $\A$ on set $\Strain$ and let $(\vv^1,\dots,\vv^k)$ be its output. \\
Let $\subspace$ be the subspace spanned by $(\vv^1,\dots,\vv^k)$ and consider the following sparse cover of $\subspace$: 
$\subspace_{\eps''} = \{\frac{\vu}{\|\vu\|_2}: \vu = \eps''\sum_{i=1}^k j_i \vv^i, j_i\in\Z\cap[-\frac{1}{\eps''},\frac{1}{\eps''}],\|\vu\|_2\neq 0\}$ \\
Let $\taus_{\eps'} = \{j\eps': j\in\Z\cap[-\frac{\threshold}{\eps'},\frac{\threshold}{\eps'}]\}$ be a cover of the candidate halfspace biases. \\
\textbf{Reject} and terminate if $\|\var_{\x\sim\Sunlabelled}(\x)\|_2 \ge 2$. \\
\For{$(\vu,\theta)\in \subspace_{\eps''}\times \taus_{\eps'}$}{
    \textbf{Reject} and terminate if $\pr_{\x\sim\Sunlabelled}[|\vu\cdot \x+\theta|\le 2\eps'^{2/3}] > 5\eps'^{2/3}$. \\
} 
Let $\F$ contain the concepts $\concept:\R^d\to\{\pm 1\}$ of the form $\concept(\x)=2\bigwedge_{i=1}^k \ind\{\vu^i\cdot \x +\theta^i\ge 0\} - 1$, where $(\vu^1,\theta^1),\dots,(\vu^k,\theta^k)\in \subspace_{\eps''}\times \taus_{\eps'}$ and $\pr_{(\x,y)\sim \Strain}[y\neq\concept(\x)] \le \eps/5$. \\
\textbf{Reject} and terminate if $\max_{\concept_1,\concept_2\in \F} \pr_{\x\sim \Stest}[\concept_1(\x)\neq \concept_2(\x)] > \eps/2$.\\
\textbf{Otherwise,} output $\hat{\concept}:\R^d\to\cube{}$ for some $\hat\concept\in \F$.
\end{algorithm2e}

\begin{proof}of \Cref{theorem:tds-general-intersections}. The proof is similar to the one of \Cref{theorem:tds-homogeneous-intersections}, but since the intersections are general, there are some additional complications. Let once more $\Strain$ be a set of $\mtrain$ samples from the training distribution, i.e., of the form $(\x,\copt(\x))$, where $\x\sim \Dtrain=\Gauss_d$ and let $\Stest$ be a set of $\mtest$ samples from the test distribution $\Dtest$. Let $C>0$ be a sufficiently large universal constant. Let $\copt:\R^d\to\cube{}$ denote the ground truth, i.e., the intersection of $k$ halfspaces
\[
    \copt(\x) = 2\wedge_{i\in[k]}\ind\{\vec \w^i\cdot \x + \tau^i \ge 0\}-1 \,,\text{ for }\w^1,\dots,\w^k\in \S^{d-1}\text{ and }\tau^1,\dots,\tau^k\in\R
\]
In the following, we will say that an event holds with high probability if it holds with probability sufficiently close to $1$ so that union bounding over all the bad events gives a probability of failure of at most $0.01$. This is possible by choosing $C$ to be a sufficiently large constant.

\paragraph{Soundness.} Suppose that the tests have accepted. We will once more use the subspace retrieval lemma from \cite{vempala2010learning}, but this time we will use a version (\Cref{lemma:subspace-retrieval-general}) that works for arbitrary halfspace intersections. We pick $\threshold = 3\sqrt{\log(10 k/\eps)}$, $\mdegree \ge \log(10k/\eps)$ and $C',C''>0$ sufficiently large universal constants.

Due to \Cref{lemma:subspace-retrieval-general}, the retrieved subspace $\subspace$ has the property that, with high probability, for any $i\in[k]$ with $\tau^i\le \threshold$ we have $\|\proj_{\subspace}\w^i\|_2 \ge 1- (\frac{\eps}{C'k})^3$, as long as $\mtrain \ge \mrec$. Consider once more $\w^i_\subspace = \frac{\proj_{\subspace}\w^i}{\|\proj_{\subspace}\w^i\|_2}$. We have $\measuredangle(\w^i,\w^{i}_\subspace) \le \frac{4\eps^{3/2}}{C'k^{3/2}}$ and for some $\vu^i\in\subspace_{\eps''}$, we have $\measuredangle(\vu^i,\w^i_\subspace)\le \frac{\eps^{3/2}}{C'k^{3/2}}$, whenever $\eps'' \le \frac{\eps^6}{6^4C'k^{7}}$ (which implies $|\subspace_{\eps''}|\le (\frac{2\cdot 6^4 C' k^{7}}{\eps^6})^k$). Therefore, for any $i\in[k]$ that corresponds to a halfspace with bounded bias $\tau^i\le \threshold$, we have $\measuredangle(\vu^i,\w^i) \le (\frac{5\eps}{C'k})^{3/2}$. Moreover, for any such $i$, there is some $\theta^i\in \T_{\eps'}$ that is either close to the $i$-th threshold ($|\theta^i-\tau^i| \le \eps'$) or they are both large enough ($\tau^i\le -\threshold$ and $\theta^i = -\threshold$). Assume without loss of generality that $\{i\in[k]: \tau^i\le \threshold\} = [\ell]$ for some $\ell \le k$.

Consider now the hypothesis $\concept(\x) = 2\wedge_{i\in[\ell]}\ind\{\vu^i\cdot \x + \theta^i\ge 0\} - 1$. Once more, it suffices to show that $\concept$ belongs in the set $\F$ of candidate concepts and that $\concept$ has small test error $\pr_{\x\sim\Stest}[\concept(\x)\neq \copt(\x)] \le \eps/4$, because then for any other candidate concept $\concept'\in\F$, we know that it disagrees with $\concept$ only on a small fraction of test points and, hence, we will have $\pr_{\x\sim\Stest}[\concept'(\x)\neq \copt(\x)]\le 3\eps/4$. By standard VC dimension arguments, this would imply that, whenever $\mtest\ge C\frac{dk\log k}{\eps^2}$, with high probability, the test error of any element of $\F$ satisfies $\pr_{\x\sim\Dtest}[\concept'(\x)\neq \copt(\x)]\le \eps$.

As a first step, we will show that the ground truth is close to the intersection corresponding to the bounded bias halfspaces with respect to both the training and the test examples, i.e., that for $\tilde\copt(\x) = 2\wedge_{i\in[\ell]}\ind\{\w^i\cdot \x + \tau^i\ge 0\}-1$ we have $\pr_{\x\sim\Stest}[\copt(\x)\neq\tilde\copt(\x)] \le \eps/8$ and $\pr_{\x\sim\Strain}[\copt(\x)\neq\tilde\copt(\x)] \le \eps/10$. This is important, because we can then relate $\concept$, $\copt$ through $\tilde\copt$. Since the moment-matching test has accepted, by \Cref{lemma:moment-matching-concentration}, as long as $\mdegree \ge \log(10k/\eps)$ and $\threshold \ge 3\sqrt{\log(10 k/\eps)}$, for any $i>\ell$, we have that $\pr_{\x\sim\Stest}[\sign(\w^i\cdot \x+\tau^i)\neq 1] \le \frac{\eps}{10k}$. Therefore, $\pr_{\x\sim\Stest}[\copt(\x)\neq \tilde\copt(\x)] \le \sum_{i>\ell}\pr_{\x\sim\Stest}[\sign(\w^i\cdot \x+\tau^i)\neq 1] \le \eps/8$, due to a union bound (and the fact that the only possibility that $\copt$ and $\tilde{\copt}$ differ is if some of the omitted halfspaces in $\tilde{\copt}$ becomes negative). Similarly, for $\Strain$, the claim follows with high probability by a standard Hoeffding bound ($\copt$ and $\tilde{\copt}$ do not depend on $\Strain$), as long as $|\Strain| \ge C\frac{k^2}{\eps^2}$.

We will now bound the quantity $\pr_{\x\sim\Stest}[\concept(\x)\neq \tilde{\copt}(\x)]$ by $\eps/8$. Observe that in the case that $|\tau^i|\ge \threshold$, then, by \Cref{lemma:moment-matching-concentration} (as argued above), the corresponding halfspace is constant with probability at least $1-\eps/(10k)$ and the same is true for $\theta^i = \threshold$. Therefore, we may safely omit these terms from $\concept$ and $\tilde{\copt}$ by only incurring an error of at most $\eps/10$. For the remaining terms, we appeal to the tester for local (general) halfspace disagreement of \Cref{lemma:general-disagreement-tester} in order to show that $\pr_{\x\sim\Stest}[\concept(\x)\neq \tilde{\copt}(\x)] \le \eps/8$. In particular, we have that
\begin{align*}
    \pr_{\x\sim\Stest}[\concept(\x)\neq \tilde{\copt}(\x)] &\le k \pr_{\x\sim\Stest}[\sign(\vu^i\cdot \x+\theta^i)\neq \sign(\w^i\cdot \x+\tau^i)] \\
    &\le C'' k(\measuredangle(\vu^i,\w^i))^{2/3} + C''k (\measuredangle(\vu^i,\w^i))\log^{1/2}(1/\eps) \\
    &\le \eps/8
\end{align*}

Finally, we show that the hypothesis $\concept$ lies within $\F$. In particular, $\pr_{\x\sim\Dtrain}[\concept(\x)\neq \tilde{\copt}(\x)] \le k\pr_{\x\sim\Gauss_d}[\sign(\vu^i\cdot \x+\theta^i)\neq\sign(\w^i\cdot \x+\tau^i)] = O( k\threshold\measuredangle(\vu^i,\w^i))$, which is bounded by $\eps/20$ by choosing the constant $C'$ appropriately. By standard VC dimension arguments, we therefore have that $\pr_{\x\sim\Strain}[\concept(\x)\neq \copt(\x)] \le \eps/5$ as long as $\mtrain \ge \frac{Cdk\log k}{\eps^2}$.

\paragraph{Completeness.} To prove completeness, suppose that $\Dtest = \Gauss_d$. Since $\subspace_{\eps''},\F$ do not depend on $\Stest$, we can use Hoeffding bounds to bound the probability of rejection, as well as union bounds over $\F\times \F$ accordingly. In particular, the tester of \Cref{lemma:general-disagreement-tester} will accept with high probability as long as $\mtest \ge C\frac{1}{\eps'^{4/3}} + Cd\log^2d = O(\frac{k^2}{\eps^2} + d\log^2d)$ and the tester of the disagreement probabilities of pairs in $\F$ will accept (due to standard Hoeffding and union bounds) with high probability whenever $\mtest \ge C \frac{1}{\eps^2} \log|\F| = O(\frac{k^2}{\eps^2} \log(\frac{k}{\eps}))$ (since $|\F| = (k/\eps)^{O(k^2)}$ as we need to choose $k$ normals from $\subspace_{\eps''}$ and $k$ elements from $\taus_{\eps'}$). For the moment matching tester, we require that $\mtest \ge Cd^{4\log(k/\eps)}$, since the tester would then have to accept with high probability (see also Lemma D.1 in \cite{klivans2023testable}). 
\end{proof}

By combining \Cref{theorem:tds-general-intersections-appendix} with \Cref{lemma:subspace-retrieval-general,lemma:subspace-retrieval-non-degenerate-appendix,lemma:subspace-retrieval-homogeneous-polar-planes} we obtain the following bounds for TDS learning general halfspace intersections.

\begin{corollary}[TDS Learning Bounds for General Halfspace Intersections]\label{corollary:general-results}
    Let $\bias\in(0,\frac{1}{2})$, $\eps>0$, $\nondegen\ge 1$ and let $\C$ be the class of intersections of $k$ general halfspaces on $\R^d$.
    \begin{enumerate}[label=\textnormal{(}\alph*\textnormal{)}]
        \item There is an $(\eps,\delta=0.02)$-TDS learner for the class $\C_\bias$ of $\bias$-balanced intersections that uses $\tilde{O}(d)(\frac{k}{\eps\bias})^{O(\frac{k^6}{\eps^6})}$ labelled training examples, $d^{O(\log(k/\eps))}$ unlabelled test examples and runs in time $\tilde{O}(d^3)(\frac{k}{\eps\bias})^{O(\frac{k^6}{\eps^6})}+d^{O(\log(k/\eps))}(k/\eps)^{O(k^2)}$.
        \item There is an $(\eps,\delta=0.02)$-TDS learner for the class $\C_\bias^\nondegen$ of $\bias$-balanced and $\nondegen$-non-degenerate intersections that uses $\tilde{O}(d)\cdot \frac{1}{\bias^2}\cdot (\frac{k}{\eps})^{O(\nondegen)}$ labelled training examples, $d^{O(\log(k/\eps))}$ unlabelled test examples and runs in time $\tilde{O}(d^3)\cdot \frac{1}{\bias^2}\cdot (\frac{k}{\eps})^{O(\nondegen)} + d^{O(\log(1/\eps))}(k/\eps)^{O(k^2)}$.
    \end{enumerate}
\end{corollary}

\section{SQ Lower Bounds for TDS Learning}\label{section:appendix-sq}

\subsection{SQ Lower Bounds for TDS Learning General Halfspaces}\label{section:appendix-sq-single-halfspaces}

In this section, we provide the proof of the SQ lower bound for TDS learning general halfspaces. Recall that the proof consists of two main steps. First, we reduce the problem of biased halfspace detection of \Cref{definition:bias-halfspace-detection} to TDS learning halfspaces and then we show that the bias halfspace detection problem is hard in the SQ framework. 

\subsubsection{Detecting Biased Halfspaces through TDS Learning}\label{section:detecting-via-tds-appendix}

For the first ingredient we use the following proposition which we restate here for convenience.

\begin{proposition}[Biased Halfspace Detection via TDS Learning]
\label{proposition: tds learning implies biased halfspace detection - appendix}
Let $\mathcal{A}$ be a TDS learning algorithm for general halfspaces
over $\R^{d}$ w.r.t. $\Gauss_d$ with accuracy parameter
$\eps$ and success probability at least $0.95$. Suppose $\mathcal{A}$
obtains at most $m$ samples from the training distribution and accesses
the test distribution via $N$ SQ queries of tolerance $\tol$
(the SQ queries are allowed to depend on the training samples). Then,
there exists an algorithm $(\frac{1}{100m},10\epsilon)$-biased halfspace
detection that uses $N+1$ SQ queries of tolerance $\min\left(\tol,\epsilon\right)$
and has success probability at least $0.8$. 
\end{proposition}

\begin{proof}
Without loss of generality, suppose that the algorithm $\mathcal{A}$
uses exactly $m$ samples from the training distribution. We use the following
algorithm that uses the TDS learning algorithm $\mathcal{A}$.
\begin{itemize}
\item \textbf{Given: }Statistical query access to distribution $\Dgeneric$ over
$\R^{d}$ with tolerance $\min\left(\phi,\epsilon\right)$.
\item \textbf{Output: }``Accept'' or ``Reject''.
\end{itemize}
\begin{enumerate}
\item Generate $S_{\text{train}}\subset\R^{d}\times\left\{ \pm1\right\} $,
of pairs $(\x^i,$-1), where each $\x^{i}$ is sampled from
$\Gauss_d$.
\item Run the TDS learning algorithm $\mathcal{A}$ on the training set
$S_{\text{train}}$. Every time $\mathcal{\mathcal{A}}$ makes an
SQ query to the test distribution, make the same SQ query to $\Dgeneric$,
and return $\mathcal{A}$ the result. 
\item If $\mathcal{A}$ returns ``Reject'', then our algorithm also returns
``Reject'' and terminates.
\item Otherwise, $\mathcal{A}$ outputs ``Accept'' and a classifier $\hat{f}:\R^{d}\to\left\{ \pm1\right\} $.
\item Using an SQ query, let $\hat{\lambda}$ be an estimate up to additive
error $\min\left(\phi,\epsilon\right)$ of $\pr_{\x\sim \Dgeneric}\left[\hat{f}(\x)=1\right]$.
\item If $\hat{\lambda}>4\epsilon$, then output ``Reject'' and terminate.
\item Otherwise, output ``Accept'' and terminate.
\end{enumerate}
First, we argue that if $\Dgeneric$ is $\Gauss_d$, then the algorithm
above will output ``Accept'' with probability at least $0.8$. For
arbitrarily chosen unit vector $\w$, as a parameter $\tau$
grows to infinity, the statistical distance between $S_{\text{train}}=\left\{ (\x^i,-1)\right\} $
and the set $S_{\text{train}}'=\left\{ (\x^i,\sign\left(\w\cdot\x^i-\tau\right))\right\} $
goes to zero. If $\mathcal{A}$ is given $S_{\text{train}}'$ and
$\Dgeneric=\Gauss_d$, then the definition of TDS learning requires
$\mathcal{A}$ with probability at least $0.95$ to accept and output
a hypothesis $\hat{f}$ satisfying $\pr_{\x\sim\Gauss_d}\left[\hat{f}(\x)\neq\sign\left(\w\cdot\x-\tau\right)\right]\leq\epsilon$.
Taking the parameter $\tau$ to be sufficiently large, we see that if
$\mathcal{A}$ is given $S_{\text{train}}=\left\{ (\x^i,-1)\right\} $
and $\Dgeneric=\Gauss_d$, then with probability at least $0.94$
the algorithm $\mathcal{A}$ accepts and outputs a hypothesis $\hat{f}$
satisfying $\pr_{\x\sim\Gauss_d}\left[\hat{f}(\x)\neq-1\right]\leq2\epsilon$.
Therefore, the estimate $\hat{\lambda}$ will be at most $3\epsilon$,
and we will thus output ``Accept''.

Now, suppose $\Dgeneric$ is such that for some unit vector $\vv$ and
$\tau\in\R$ we have $\pr_{\x\sim \Dgeneric}[\x\cdot\vv\geq \tau]\geq10\epsilon$
and $\pr_{\x\sim\Gauss_d}[\x\cdot\vv\geq \tau]\leq\frac{1}{100m}$.
Besed on the set $S_{\text{train}}=\left\{ (\x^i,-1)\right\} $,
define the set $S_{\text{train}}''$ as $S_{\text{train}}''=\left\{ (\x^i,\sign\left(\vv\cdot\x^i-\tau\right))\right\} $.
If the algorithm $\mathcal{A}$ were given the set $S_{\text{train}}''$
instead of $S_{\text{train}}$ as the training set, then the definition
of TDS learning would require $\mathcal{\mathcal{A}}$ with probability
at least $0.95$ either to output ``Reject'' or give a hypothesis
$\hat{f}$ satisfying $\pr_{\x\sim \Dgeneric}\left[\hat{f}(\x)\neq\sign\left(\vv\cdot\x-\tau\right)\right]\leq\epsilon$.
Since $\pr_{\x\sim\Gauss_d}[\x\cdot\vv\geq \tau]\leq\frac{1}{100m}$
and $\left|S_{\text{train}}\right|=\left|S_{\text{train}}''\right|=m$,
we see via a union bound that that the statistical distance between
$S_{\text{train}}$ and $S_{\text{train}}''$ is at most $0.01$.
Thus, in the algorithm above, the algorithm $\mathcal{\mathcal{A}}$
with probability at least $0.94$ indeed either outputs ``Reject''
or gives a hypothesis $\hat{f}$ satisfying $\pr_{\x\sim \Dgeneric}\left[\hat{f}(\x)\neq\sign\left(\vv\cdot\x-\tau\right)\right]\leq\epsilon$.
In the former case, our algorithm will also output ``Reject''. In
the latter case we will have $\hat{\lambda}>9\epsilon$, since $\Dgeneric$
is such that $\pr_{\x\sim \Dgeneric}[\x\cdot\vv\geq \tau]\geq10\epsilon$.
Therefore, in this case too our algoirthm outputs ``Reject'', which
completes the proof.
\end{proof}

\subsubsection{Lower Bounds for Detecting Biased Halfspaces}

We now provide a proof for the second ingredient, namely, that no efficient SQ algorithm can solve the problem of detecting biased halfspaces, i.e., the following proposition (restated here for convenience).

\begin{proposition}[SQ Lower Bounds for Biased Halfspace Detection]
\label{proposition: detecting biased halfspaces is hard - appendix} For $\epsilon>0$,
set $d=\frac{1}{\epsilon^{1/4}}$. Then, for all
sufficiently small $\epsilon$, the following is true. Suppose $\mathcal{A}$
is an SQ algorithm for $(d^{-\ln(1/\epsilon)},10\epsilon)$-biased
halfspace detection problem over $\R^{d}$, and $\mathcal{A}$ has
a success probability of at least $2/3$. Then, $\mathcal{A}$ either
has to use SQ tolerance of $d^{-\Omega(\frac{\log1/\epsilon}{\log\log1/\epsilon})}$,
or make $2^{d^{\Omega(1)}}$ SQ queries.
\end{proposition}

To prove the above claim, we first construct a one-dimensional
distribution $\Dgeneric_{1}$ that approximately matches the low-degree moments of $\Gauss_d$, while having a lot of probability mass above a certain threshold.

\begin{proposition}
\label{proposition: can match low-degree moments approximately }
For $\epsilon>0$,
let $k_{0}$ be defined as $k_{0}=\frac{\ln1/\epsilon}{100\ln\ln1/\epsilon}$.
If $\epsilon$ is sufficiently small, then there exists a distribution
$\Dgeneric_{1}$ supported on a finite subset of $\R$, satisfying 
\[
\left|\E_{x\sim \Dgeneric_{1}}\left[x^{i}\right]-\E_{x\sim\Gauss_1}\left[x^{i}\right]\right|\leq\frac{1}{k_{0}^{10k_{0}}},
\]
for every $i\in\left\{ 0,\cdots,10k_{0}\right\} $ while also satisfying
$\pr_{x\sim \Dgeneric_{1}}[x\geq t]\geq12\epsilon,$ for some $t$ for which
$\pr_{x\sim\Gauss_1}[x\geq t]\leq\epsilon^{\frac{1}{4}\ln1/\epsilon}.$
\end{proposition}

\begin{proof}
We will first construct a distribution $\Dgeneric_{1}'$ that satisfies the
conditions above, but does not have finite support. Afterwards, we
will discretize $\Dgeneric_{1}'$.

We take $t:=\ln1/\epsilon$ and observe that 
\begin{align}
\pr_{x\sim\Gauss_1}[x\geq t]&=\frac{1}{\sqrt{2\pi}}\int_{\ln1/\epsilon}^{\infty}e^{-x^{2}/2}\;d x\leq\frac{e^{-\left(\ln1/\epsilon\right)^{2}/2}}{\sqrt{2\pi}}\int_{0}^{\infty}e^{-x\ln1/\epsilon}\;d x \nonumber\\
&=\underbrace{\frac{e^{-\left(\ln1/\epsilon\right)^{2}/2}}{\sqrt{2\pi}\ln1/\epsilon}\leq\epsilon^{\frac{1}{4}\ln1/\epsilon}}_{\text{For \ensuremath{\epsilon} sufficiently small.}}.\label{eq: gaussian tail}
\end{align}
Let $\tau$ be the real number for which $\pr_{x\sim\Gauss_1}\left[x\in\left[0,\tau\right]\right]=13\epsilon$.
From Equation \ref{eq: gaussian tail}, we see that for all sufficiently
small $\epsilon$ it is the case that $\tau<\epsilon$. We define
$\Dgeneric_{1}'$ the following way: to sample $z\sim \Dgeneric_{1}'$ (i) sample $x\sim\Gauss_1$
(ii) if $x\in[0,\tau]$, then $z=t$ (iii) otherwise, $z=x$. Since
$\pr_{x\sim\Gauss_1}\left[x\in\left[0,\tau\right]\right]=13\epsilon$,
we see that $\pr_{x\sim \Dgeneric_{1}'}\left[x\geq t\right]\geq13\epsilon.$
Furthermore, we see that for every $i\in\left\{ 0,\cdots,10k_{0}\right\} $
\begin{multline*}
\left|\E_{x\sim \Dgeneric_{1}'}\left[x^{i}\right]-\E_{x\sim\Gauss_1}\left[x^{i}\right]\right|\leq t^{k_{0}}\pr_{x\sim\Gauss_1}\left[x\in\left[0,\tau\right]\right]=12\epsilon\cdot(\ln1/\epsilon)^{\frac{\ln1/\epsilon}{100\ln\ln1/\epsilon}}=\\
=12\epsilon^{0.99}\underbrace{\leq\frac{1}{2}\cdot\left(\frac{100\ln\ln1/\epsilon}{\ln1/\epsilon}\right)^{\frac{\ln1/\epsilon}{10\ln\ln1/\epsilon}}}_{\text{For \ensuremath{\epsilon} sufficiently small.}}=\frac{1}{2k_{0}^{10k_{0}}}.
\end{multline*}

Overall, we have so far shown that $\pr_{x\sim \Dgeneric_{1}'}\left[x\geq t\right]\geq13\epsilon$
and $\left|\E_{x\sim \Dgeneric_{1}'}\left[x^{i}\right]-\E_{x\sim\Gauss_1}\left[x^{i}\right]\right|<\frac{1}{2k_{0}^{10k_{0}}}$,
but $\Dgeneric_{1}'$ is not supported on a finite subset of $\R$. We will
now construct a finitely-supported distribution $\Dgeneric_{1}$ via the probabilistic
method. Obtain $\Dgeneric_{1}$ as the empirical distribution over $K$ i.i.d.
samples from $\Dgeneric_{1}'$. Since all moments of $\Dgeneric_{1}'$ are bounded,
as $K$ grows to infinity, for all $i\in\left\{ 0,\cdots,10k_{0}\right\} $
the quantity $\E_{x\sim \Dgeneric_{1}}\left[x^{i}\right]$ converges in probability
to $\E_{x\sim \Dgeneric_{1}'}\left[x^{i}\right]$, and the quantity $\pr_{x\sim \Dgeneric_{1}}\left[x\geq t\right]$
converges in probility to $\pr_{x\sim \Dgeneric_{1}'}\left[x\geq t\right]$.
Thus, for a sufficiently large $K$, we have $\pr_{x\sim \Dgeneric_{1}}\left[x\geq t\right]\geq12\epsilon$
and $\left|\E_{x\sim \Dgeneric_{1}}\left[x^{i}\right]-\E_{x\sim\Gauss_1}\left[x^{i}\right]\right|<\frac{1}{k_{0}^{10k_{0}}}$, with non-zero probability over
the choice of $\Dgeneric_{1}$,
which completes the proof.
\end{proof}

We now apply the following theorem which is implicit in \cite{pmlr-v195-diakonikolas23b-sq-mixtures} to obtain a distribution $\Dgeneric$
over $\R$ that has a lot of probability mass above a certain threshold
and whose moments match $\Gauss_1$ exactly. 

\begin{theorem}[Implicit in \cite{pmlr-v195-diakonikolas23b-sq-mixtures}]
\label{theorem: from approx to exact mom match}Let $k$ be a sufficiently
large positive integer and let $\Dgeneric_{0}$ be a distribution supported
on a finite subset of $\R$, and suppose that for every $i\in\left\{ 0,\cdots,10k\right\} $
we have 
\begin{equation}
\left|\E_{x\sim \Dgeneric_{0}}\left[x^{i}\right]-\E_{x\sim\Gauss_1}\left[x^{i}\right]\right|\leq\frac{1}{k^{10k}},\label{eq: D0 has moments close to Gaussian}
\end{equation}
then there exists a distribution $\Dgeneric_{1}$ with the same support as
$\Dgeneric_{0}$ with $\E_{x\sim \Dgeneric_{0}}\left[x^{i}\right]=\E_{x\sim\Gauss_1}\left[x^{i}\right]$
for every $i\in\left\{ 0,\cdots,k\right\} $, and also satisfying
\[
\pr_{x\sim \Dgeneric_{1}}[x=x_{0}]\geq0.9\pr_{x\sim \Dgeneric_{0}}[x=x_{0}]
\]
for every $x_{0}$ in the support of $\Dgeneric_{0}$.
\end{theorem}

The proof is equivalent to the proof given by \cite{pmlr-v195-diakonikolas23b-sq-mixtures}, but is provided here with slight modifications for completeness. We will need the following fact.

\begin{fact}
\label{fact: normalized polynomials have not too large coefficients}Let
$p$ be a polynomial over $\R$ of degree at most $k$, and let $\E_{x\sim\Gauss_1}\big[\big(p(x)\big)^{2}\big]\leq1$.
Then, each coefficient of $p$ has absolute value of at most $2^{k+1}$.
\end{fact}

\begin{proof}
We will use the Hermite polynomials. Recall that for $i=0,1,2,\cdot$ Hermite polynomials $\{H_i\}$ are the unique collection of polynomials over $\R$ that are orthogonal with respect to Gaussian distribution. In other words $\E_{x \in \Gauss_1}[H_i(x) H_j(x)]=0$ whenever $i\neq j$. In this work, we normalize the Hermite polynomials to further satisfy $\E_{x \in \Gauss_1}[H_i(x) H_i(x)]=1$. It is a standard fact from theory of orthogonal polynomials that $H_0(x)=1$, $H_1(x)=x$ and for $i\geq 2$ Hermite polynomials satisfy the following recursive identity:
	\[
	H_{i+1}(x) \cdot \sqrt{(i+1)!}
	=
	x H_i(x)\cdot \sqrt{i!} 
	-
	i \cdot
	H_{i-1}(x)\cdot \sqrt{(i-1)!}
	\]
	It follows immediately from the recursion relation that
		Each coefficient of $H_i$ is bounded by $2^i$ in absolute value.
	We expand $P(x)$ as a sum of Hermite polynomials\footnote{Note that the expansion below is always possible for a degree $k$ polynomial because polynomials of the form $H_{i}$ have degree at most $k$ and are linearly independent, because they are orthonormal with respect to the standard Gaussian distribution.}:
		\begin{equation}
			\label{eq: multi-dimentional Hermit expansion}
			p(x)
			=
			\sum_{
				i=0
			}^k
			\alpha_{i}
			H_{i}(x)
		\end{equation}
		Due to orthogonality of Hermite polynomials, we have:
		\[
		\sum_{
				i=0
			}^k
		\alpha_{i}^2
		=
		\E_{x \in \mathcal{N}(0, 1)}[(p(x))^2]
		\leq 1    
		\]
		In particular, this implies that each coefficient $ \alpha_{i}$ is bounded by $1$ in absolute value. Combining this with Equation \ref{eq: multi-dimentional Hermit expansion}, the fact that each coefficient of $H_i$ is bounded by $2^i$ in absolute value, we see that each coefficient of $p$ is bounded by $\sum_{i=0}^k 2^i<2^{k+1}$ in absolute value.
\end{proof}

\begin{proof}of \Cref{theorem: from approx to exact mom match}, implicit in \cite{pmlr-v195-diakonikolas23b-sq-mixtures}. Provided here for completeness. 

We first restate the setting of the theorem.
Let $k$ be a sufficiently large positive integer and let $\Dgeneric_{0}$
be a distribution supported on a finite subset of $\R$, and suppose
that for every $i\in\left\{ 0,\cdots,10k\right\} $ we have 
\begin{equation}
\left|\E_{x\sim \Dgeneric_{0}}\left[x^{i}\right]-\E_{x\sim\Gauss_1}\left[x^{i}\right]\right|\leq\frac{1}{k^{10k}},\label{eq: D0 has moments close to Gaussian-1}
\end{equation}
then we would like to show that there exists a distribution $\Dgeneric_{1}$
with the same support as $\Dgeneric_{0}$ satisfying $\E_{x\sim \Dgeneric_{0}}\left[x^{i}\right]=\E_{x\sim\Gauss_1}\left[x^{i}\right]$
for every $i\in\left\{ 0,\cdots,k\right\} $, and also satisfying
\[
\Pr_{x\sim \Dgeneric_{1}}[x=x_{0}]\geq0.9\Pr_{x\sim \Dgeneric_{0}}[x=x_{0}]
\]
for every $x_{0}$ in the support of $\Dgeneric_{0}$.

Let $N$ denote the number of elements in the support
of $\Dgeneric_{0}$ and let $\{x_{1},\cdots, x_{N}\}$ be the elements in the
support of $\Dgeneric_{0}$. Consider the following linear program:

\begin{align*}
\text{Find } &  & \mu_{x_{1}},\cdots\mu_{x_{N}}\\
s.t. &  & \E_{x\sim \Dgeneric_{0}}\left[\mu_{x}p(x)\right]=\E_{x\sim\Gauss_1}\left[p(x)\right] &  & \text{for every polynomial \ensuremath{p} of degree at most \ensuremath{k}}\\
 &  & \mu_{x_{j}}\geq0.9 &  & \text{for all \ensuremath{j\in\left\{ 1,\cdots N\right\} }}
\end{align*}
If the linear program above is feasible, then the proposition will
be satisfied by a distribution $\Dgeneric_{1}$ supported on ${x_{1},\cdots x_{N}}$
that has probability $\mu_{x_{j}}\Pr_{x\sim \Dgeneric_{0}}\left[x=x_{j}\right]$
on each $x_{j}$ (note that $\Dgeneric_{1}$ is indeed a probability distribution
because the equality $\sum_{j}\mu_{x_{j}}\Pr_{x\sim \Dgeneric_{0}}\left[x=x_{j}\right]=1$
follows by the constraint in the linear program when $p$ is identically
equal to $1$).

The linear program above is feasible if and only if its dual linear
program is infeasible. The dual linear program is as follows:

\begin{align}
\text{Find } &  & \text{polynomial }p\text{ of degree at most \ensuremath{k},}\label{eq: dual linear program}\\
s.t. &  & p(x_{j})\geq0 &  & \text{for all }j\in\left\{ 1,\cdots,N\right\} ,\nonumber \\
 &  & \E_{x\sim\Gauss_1}\left[p(x)\right]<0.9\E_{x\sim \Dgeneric_{0}}\left[p(x)\right]. &  & \text{}\nonumber 
\end{align}
It is now shown that the above is indeed infeasible if $\Dgeneric_{0}$ is
such that for every $i\in\left\{ 0,\cdots,10k\right\} $ we have$\left|\E_{x\sim \Dgeneric_{0}}\left[x^{i}\right]-\E_{x\sim\Gauss_1}\left[x^{i}\right]\right|\leq\frac{1}{k^{10k}}$.
For the sake of contradiction, suppose that the linear program above
is feasible and is satisfied by some polynomial $p$. Without loss
of generality, assume that $\E_{x\sim\Gauss_1}\left[\left(p(x)\right)^{2}\right]=1$,
because otherwise one could rescale $p$ while still satisfying the
dual linear program above. By Fact \ref{fact: normalized polynomials have not too large coefficients}
each coefficient of $p$ has absolute value of at most $2^{k+1}$. This
implies that each coefficient of $p^{2}$ has an absolute value of
at most $8^{k+1}$ and each coefficient of $p^{4}$ has an absolute
value of at most $32^{k+1}$. Combining these coefficient bounds with
Equation \ref{eq: D0 has moments close to Gaussian-1}, and applying
the triangle inequality, we see that 
\begin{align}
\left|\E_{x\sim \Dgeneric_{0}}\left[p(x)\right]-\E_{x\sim\Gauss_1}\left[p(x)\right]\right| & \leq\frac{(k+1)2^{k+1}}{k^{10k}},\label{eq: p has good exp}\\
\left|\E_{x\sim \Dgeneric_{0}}\left[\left(p(x)\right)^{2}\right]-\E_{x\sim\Gauss_1}\left[\left(p(x)\right)^{2}\right]\right| & \leq\frac{(2k+1)8^{k+1}}{k^{10k}},\label{eq: p2 has good exp}\\
\left|\E_{x\sim \Dgeneric_{0}}\left[\left(p(x)\right)^{4}\right]-\E_{x\sim\Gauss_1}\left[\left(p(x)\right)^{4}\right]\right| & \leq\frac{(4k+1)32^{k+1}}{k^{10k}}.\label{eq: p4 has good exp}
\end{align}
This allows us to upper-bound $\E_{x\sim \Dgeneric_{0}}\left[\left|p(x)\right|\right]$
as follows, where the first inequality follows by Equation \ref{eq: p has good exp}, the second by the fact that $p$ satisfies the Linear Program \ref{eq: dual linear program} and the equality because $p$ is positive on the support of $\Dgeneric_{0}$ due satisfying the Linear Program \ref{eq: dual linear program}.

\begin{equation}
{\frac{(k+1)2^{k+1}}{k^{10k}}\geq\E_{x\sim \Dgeneric_{0}}\left[p(x)\right]-\E_{x\sim\Gauss_1}\left[p(x)\right]}{>0.1\E_{x\sim \Dgeneric_{0}}\left[p(x)\right]}{=0.1\E_{x\sim \Dgeneric_{0}}\left[\left|p(x)\right|\right]}\label{eq: exp of p small}
\end{equation}
However, we can also lower-bound $\E_{x\sim \Dgeneric_{0}}\left[\left|p(x)\right|\right]$
in the following way

\begin{align}
\E_{x\sim \Dgeneric_{0}}\left[\left|p(x)\right|\right] &\overbrace{\geq\frac{\left(\E_{x\sim \Dgeneric_{0}}\left[\left(p(x)\right)^{2}\right]\right)^{3/2}}{\E_{x\sim \Dgeneric_{0}}\left[\left(p(x)\right)^{4}\right]}}^{\text{By generalized Holder inequality.}}\overbrace{\geq\frac{\left(\E_{x\sim\Gauss_1}\left[\left(p(x)\right)^{2}\right]-\frac{(2k+1)8^{k+1}}{k^{10k}}\right)^{3/2}}{\E_{x\sim\Gauss_1}\left[\left(p(x)\right)^{4}\right]+
\frac{(4k+1)32^{k+1}}{k^{10k}}
}}^{\text{By Equations \ref{eq: p2 has good exp} and \ref{eq: p4 has good exp}.}}\geq \nonumber\\
&{\geq\frac{\left(1-\frac{(2k+1)8^{k+1}}{k^{10k}}\right)^{3/2}}{(4k+1)32^{k+1}k!!+\frac{(4k+1)32^{k+1}}{k^{10k}}}}{\geq\frac{1}{k^{k}}},\;{\text{ for sufficiently large \ensuremath{k}.}}\label{eq: exp of p large}
\end{align}
where the prior to last inequality follows form the fact that $\E_{x\sim\Gauss_1}\left[\left(p(x)\right)^{4}\right]\leq(4k+1)32^{k+1}k!!$, as each coefficient of $p^{4}$ is at most $32^{k+1}$ in absolute value.
Overall, we see that Equations \ref{eq: exp of p large} and \ref{eq: exp of p small}
cannot hold simultaneously for a sufficiently large $k$ , contradiction.
\end{proof}

In order to conclude the proof of \Cref{proposition: detecting biased halfspaces is hard - appendix}, we a tool from \cite{diakonikolas2023sq}.

\begin{theorem}[{Special case of \cite{diakonikolas2023sq}}]
\label{theorem: moment matching implies lowe bound} 
Let $\Dgeneric$ be a distribution
over $\R$ such that for every $i\in\left\{ 0,\cdots,k\right\} $
we have $\E_{x\sim \Dgeneric}\left[x^{i}\right]=\E_{x\sim\Gauss_1}\left[x^{i}\right]$.
For a unit vector $\vv$, let $\Dgeneric_{\vv}$ denote the distribution
over $\R^{d}$ such that for $\x\sim \Dgeneric_{\vv}$ (i) the projection
$\x\cdot\vv$ is distributed as $\Dgeneric$ (ii) the projection
of $\x$ onto the subspace orthogonal to $\vv$ is distributed
as $\Gauss_{d-1}$ independently from $\x\cdot\vv$.
Suppose $\mathcal{A}$ is an SQ algorithm that distinguishes with
success probability at least $2/3$ the distribution $\Gauss_d$
from $\Dgeneric_{\vv}$, with $\vv$ a uniformly random unit vector.
Then, $\mathcal{A}$ either needs to use SQ tolerance of $k^{10k}d^{-0.1k}$
or make $2^{d^{\Omega(1)}}$ SQ queries.
\end{theorem}

\subsubsection{TDS Learning General Halfspaces is Hard for SQ Algorithms}

Finally, we prove Theorem \ref{theorem: TDS learning requires d to the log}
by combining the reduction of Proposition \ref{proposition: tds learning implies biased halfspace detection}
with the SQ lower bound of Proposition \ref{proposition: detecting biased halfspaces is hard}
to obtain an SQ lower bound for TDS learning of general halfspaces. 

Recall that in the setting of Theorem \ref{theorem: TDS learning requires d to the log}
for $\epsilon>0$, we let $d$ be chosen as $d=\frac{1}{\epsilon^{1/4}}$.
Suppose Theorem \ref{theorem: TDS learning requires d to the log} is
false. Then for a sequence of $\epsilon$ approaching $0$ there is
a TDS learning algorithm $\mathcal{A}$ for general halfspaces over
$\R^{d}$ with accuracy parameter $\epsilon$ and success probability
at least $0.95$. The algorithm $\mathcal{A}$ obtains at most $d^{\frac{\log1/\epsilon}{\log\log1/\epsilon}}$
samples from the training distribution and accesses the testing distribution
via $2^{d^{o(1)}}$ SQ querries of tolerance at least $d^{-o(\frac{\log1/\epsilon}{\log\log1/\epsilon})}$.

Combining this with Proposition \ref{proposition: tds learning implies biased halfspace detection},
we see that for an infinite sequence of values of positive $\epsilon$
that approaches zero, there exists an algorithm for $(\frac{1}{100}d^{-\frac{\log1/\epsilon}{\log\log1/\epsilon}},10\epsilon)$-biased
halfspace detection that uses $2^{d^{o(1)}}$ SQ querries of tolerance
$\min(d^{-o(\frac{\log1/\epsilon}{\log\log1/\epsilon})},\epsilon)=d^{-o(\frac{\log1/\epsilon}{\log\log1/\epsilon})}$
and has success probability at least $0.8$. However, for sufficiently
small values of $\epsilon$, this directly contradicts Proposition
\ref{proposition: detecting biased halfspaces is hard}. This finishes the
proof of Theorem \ref{theorem: TDS learning requires d to the log}.

\subsection{SQ Lower Bounds for Intersections of two Homogeneous Halfspaces}\label{section:appendix-hardness-two-homogeneous}

In order to prove \Cref{theorem:sq-two-homogeneous}, it suffices to reduce the anti-concentration detection problem of \Cref{theorem:sq-hardness-ac} to TDS learning of two homogeneous halfspaces.

The reduction follows the template of the proof of \Cref{proposition: tds learning implies biased halfspace detection}. In this case, we construct a distringuisher for the AC detection problem (between the two options (1) $\Gauss_d$ and (2) $\Dgeneric'$ described in \Cref{theorem:sq-hardness-ac}) by providing training examples of the form $(\x,-1)$, $\x\sim\Gauss_d$ to the input of the TDS algorithm and the SQ oracle for the unknown distribution as an oracle to the test marginal. 

The training data are with high probability consistent with the intersection of the halfspaces $H_1 = \{\x: (\sqrt{\alpha}\vu+\sqrt{1-\alpha}\vv)\cdot \x \ge 0\}$ and $H_2 = \{\x: (\sqrt{\alpha}\vu-\sqrt{1-\alpha}\vv)\cdot \x \ge 0\}$, where $\vv,\vu\in\S^{d-1}$, $V=\{\x:\vv\cdot \x = 0\}$ is the subspace where $\Dgeneric'$ assigns non-negligible mass, $\vu \perp \vv$ and $\alpha\in(0,1/2)$ is arbitrarily small (even exponentially in $d,\frac{1}{\eps}$). Assume, also, that the mass, under $\Dgeneric'$, of $V\cap\{\x:\vu\cdot\x\ge 0\}$ is greater than the mass of $V\cap\{\x:\vu\cdot\x< 0\}$ (otherwise, note that the training data are also consistent w.h.p. with the intersection of the complement of $H_1$ with the complement of $H_2$).

Suppose that the TDS algorithm rejects. Then, we have a certificate that the test data are not Gaussian and therefore we are in the case (2) of the distinguishing problem (w.h.p.). If the TDS algorithm accepts and outputs some hypothesis $\hat f$, then we query $\pr[\hat f(\x)=1]$ to the SQ oracle for the test marginal. If the test marginal was the Gaussian, then the value of the query should be very close to $0$ (because, upon acceptance, $\hat f$ achieves low error). If the test marginal was $\Dgeneric'$, then the value of the query should be bounded away from $0$, because $\Dgeneric'$ assigns non-negligible mass to the positive region of the intersection and $\hat f$ must achieve low error. Hence, the value of the query indicates the answer to the distinguishing problem.

\subsection{SQ Lower Bounds under Non-Degeneracy Condition}\label{section:appendix-hardness-under-non-degeneracy}

In \Cref{section:appendix-recovery-non-degenerate} we define a non-degeneracy condition (\Cref{definition:non-degeneracy}) which is sufficient to obtain an exponential improvement for the problem of approximately retrieving the relevant subspace (see \Cref{lemma:subspace-retrieval-non-degenerate-appendix}). This implies improved performance for our TDS learners for halfspace intersections. Importantly, our SQ lower bounds (\Cref{theorem:sq-two-homogeneous,theorem: TDS learning requires d to the log even balanced}) hold even for under the non-degeneracy condition and this enables us to compare our upper and lower bounds under this condition.

For \Cref{theorem:sq-two-homogeneous}, the unknown intersection of the hard construction is non-degenerate, because it corresponds to an intersection of two halfspaces with normals $\w_1,\w_2$ such that $\w_1,\w_2$ are pointing almost in opposite directions. This implies that after projecting $\w_2$ on the subspace orthogonal to $\w_1$, we obtain a direction $\w'$ such that the halfspace $\{\x:\w'\cdot \x\ge 0\}$ is consistent with all of the points in the interior of the unknown intersection and therefore, by \Cref{lemma:variance-reduction}, there is significant variance reduction in the direction of $\w'$. Overall, the constructed intersection is $2$-non-degenerate.

For \Cref{theorem: TDS learning requires d to the log even balanced}, the construction corresponds to an intersection of two halfspaces with normals $\w_1,\w_2$ such that $\w_1,\w_2$ are pointing (w.h.p. as $d$ increases) in almost orthogonal directions. In this case, we do not apply \Cref{lemma:variance-reduction} directly, because the statement is not tight when the residual vector $\vu = \frac{\w_2-\proj_{\w_1}\w_2}{\|\w_2-\proj_{\w_1}\w_2\|_2}$ is very close to $\w_2$. Instead, we refer to the proof of \Cref{lemma:variance-reduction}, which implies that, if $\vu\cdot \w_2$ is sufficiently close to $1$, then we have variance reduction along $\vu$ that indeed scales proportionally to the variance reduction along $\w_2$ and, hence, the corresponding intersection is $2$-non-degenerate.

\end{document}